\documentclass[journal]{IEEEtran}

 \usepackage[utf8]{inputenc}
\usepackage[T1]{fontenc}
\usepackage{graphicx,epsf}
\usepackage{graphbox}
\usepackage{bbm}
\usepackage{amsmath,amssymb,amsfonts}
\usepackage{multirow}
\usepackage{bm}
\usepackage{algorithmic}
\usepackage[ruled,vlined,linesnumbered]{algorithm2e}
\usepackage{hyperref}
\usepackage{amsthm}
\usepackage{dsfont}
\usepackage{subfig}
\usepackage{nccmath}
\usepackage{color,xcolor}
\usepackage{cite}

\newcommand{\ZZ}{{\mathbb Z}}

\newcommand{\integer}{\mathbb{Z}}
\newcommand{\real}{\mathbb{R}}

\newcommand{\E}{\mathbb{E}}

\newcommand{\bbR}{\mathbb{R}}
\newcommand{\bbN}{\mathbb{N}}
\newcommand{\HH}{\mathbf{H}}
\newcommand{\Hv}{\underline{H}}

\usepackage[ruled,vlined,linesnumbered]{algorithm2e}

\newtheorem{theorem}{Theorem}[section]

\newtheorem{remark}[theorem]{Remark}

\newcommand{\bbP}{{\Bbb P}}

\newcommand{\bbE}{{\Bbb E}}

\newcommand{\Var}{\textnormal{Var\hspace{0.5mm}}}

\newcommand{\Cov}{\textnormal{Cov}}

\newcommand{\diag}{\textnormal{diag}}

\DeclareMathOperator*{\plim}{\mathit{p}-lim}

\graphicspath{
{./FiguresTSP/}
}

\begin{document}

\title{Multivariate selfsimilarity:
Multiscale eigen-\\ structures for selfsimilarity parameter estimation}

\author{Charles-G\'erard Lucas$^{(1)}$, Gustavo Didier$^{(2)}$, Herwig Wendt$^{(3)}$, \IEEEmembership{IEEE Senior},
Patrice Abry$^{(1)}$,~\IEEEmembership{IEEE Fellow},\thanks{Supported by  PhD Grant DGA/AID (no 01D20019023) and ANR-18-CE45-0007 MUTATION and the Simons Foundation collaboration grant $\#714014$. }
\\
\makebox{} \\
$^{(1)}$  ENSL, CNRS, Laboratoire de physique, F-69342 Lyon, France, {\tt \{patrice.abry,charles.lucas\}@ens-lyon.fr}\\
$^{(2)}$ Math. Dept., Tulane University, New Orleans, USA, {\tt gdidier@tulane.edu}, \\
$^{(3)}$  IRIT, Univ. Toulouse, CNRS, Toulouse, France, {\tt herwig.wendt@irit.fr}.
}

\maketitle

\begin{abstract}
Scale-free dynamics, formalized by selfsimilarity, provides a versatile paradigm massively and ubiquitously used to model temporal dynamics in real-world data.
However, its practical use has mostly remained univariate so far.
By contrast, modern applications often demand multivariate data analysis. Accordingly, models for multivariate selfsimilarity were recently proposed. Nevertheless, they have remained rarely used in practice because of a lack of available reliable estimation procedures for the vector of selfsimilarity parameters. Building upon recent mathematical developments, the present work puts forth an efficient estimation procedure {based on} the theoretical study of the multiscale eigenstructure of the wavelet spectrum of multivariate selfsimilar processes. The estimation performance is studied theoretically in the asymptotic limits of large {scale and} sample sizes, and computationally for finite-size samples. As a practical outcome, a fully operational and documented multivariate signal processing estimation  toolbox is made freely available and is ready for practical use on real-world data.
Its potential benefits are illustrated in epileptic seizure prediction from multi-channel EEG data.
\end{abstract}


\IEEEpeerreviewmaketitle

\section{Introduction}
\label{sec.intro}

\noindent {\bf Context.}
The concept of scale-free dynamics results from a change in paradigm: instead of being driven by a few characteristic time scales, thus targeted in estimation, temporal dynamics results from the interactions of a large continuum of time scales, hence all of equal relevance. Scale-free dynamics has been massively used in the past decades to successfully model temporal activities in a large variety of data stemming from several real-world applications very different in nature, ranging from natural phenomena --- physics (hydrodynamic turbulence \cite{Frisch1995}, astrophysics \cite{davis1992cosmic}), geophysics (rainfalls, \cite{hubert2002multifractal}, seismology \cite{turcotte1990implications}), biology (plant growths \cite{palmer1988fractal}, body rhythms \cite{Lopes2009}, heart rate \cite{Ivanov2007,Nakamura2016}, neurosciences \cite{He2014,la2018self}, genomics \cite{takahashi1989fractal}) --- to human activity --- Internet traffic \cite{fontugne2017scaling}, finance \cite{mandelbrot99}, city growth \cite{lengyel2023roughness}, art investigations \cite{AbrySPM2015}.

However, scale-free dynamics has so far remained restricted to univariate settings only, mostly due to a lack of multivariate models and estimation tools.
By contrast, in most 
applications, one same system is monitored by several sensors, which naturally calls for the joint (or multivariate) analysis of a collection of time series. Multivariate selfsimilarity with the practical and accurate estimation of the vector of selfsimilarity parameters thus constitutes the heart of the present work. \\
\noindent {\bf Related work.} Fractional Brownian motion (fBm)  \cite{Samorodnitsky1994,pipiras:taqqu:2017}, the only Gaussian, selfsimilar, stationary-increment process, has quasi-exclusively been used in practice as the reference model for scale-free dynamics.
The selfsimilarity parameter $H$ is the quantity of interest in applications. Its accurate estimation has thus received considerable attention (cf.\ \cite{Bardet2003a,pipiras:taqqu:2017} for reviews). Notably, multiscale (wavelet) representations have proven to yield accurate and reliable estimation procedures for $H$ \cite{wendt2017multivariate}: They rely on the crucial facts that the statistics of the wavelet coefficients behave as power laws with respect to scales, and that the exponents of these power laws are controlled by $H$. To account for the multivariate nature of data in applications, multivariate extensions of fBm were recently proposed based on gathering correlated fBm \cite{Amblard_P-0_2011_j-ieee-tsp_imfbm}.
Multivariate multiscale-based representations were accordingly devised for the estimation of the resulting vector of selfsimilarity parameters. Nonetheless, these constructs were essentially based on component-wise univariate power law behavior \cite{Coeurjolly_J-F_2013_ESAIM_wamfbm,wendt2017multivariate}.
More recently, a richer model for multivariate selfsimilarity, operator fractional Brownian motion (ofBm), was proposed. Its
versatility stems from a major change of perspective: there is no longer a one-to-one association either between each selfsimilarity parameter and each component, or between power laws and scale-free dynamics.
Instead, the statistics of each component depend on the entire vector of selfsimilarity parameters \cite{mason:xiao:2002,Didier_G_2011_bernouilli_irpofbm,didier:pipiras:2011:exp}.
Very recently, the principles of the estimation for the full vector of selfsimilarity parameters were theoretically discussed, based on the eigenvalues of multivariate wavelet spectra estimated at each scale \cite{abry2018waveleta,abry2018waveletb}.
This paved the way for preliminary explorations  \cite{lucas2021bootstrap} that lead to the core contributions of the present work:
 The construction of practical statistical signal processing tools for multivariate selfsimilarity analysis. \\
\noindent {\bf Goals, contributions and outline.} {The present work aims to define a reliable and theoretically well-grounded estimation procedure for the vector of selfsimilarity parameters in a multivariate setting  
and to establish its good performance both in theory and in practice.} 
To that end, Section~\ref{sec.OFBM} recalls state-of-the-art multivariate selfsimilarity modeling and analysis. As a first contribution, a {tractable} 
yet richer model for multivariate selfsimilarity is devised in Section~\ref{sec.MFBM}, and then compared to ofBm; its multivariate wavelet analysis is detailed.
As a second and major contribution, an estimation procedure for the vector of selfsimilarity parameters is devised, by exploiting a bias-variance trade-off, and its asymptotic performance is theoretically studied (cf. Section~\ref{sec.WOFBM}).
Finite-size estimation performance is studied empirically (from Monte Carlo simulations), together with the derivation of approximations to the  covariance structure of the estimates, which is of critical importance for real-world applications (cf. Section~\ref{sec.MC}).
The estimation performance of the proposed multivariate estimator is compared, in terms of biases, variances, MSE, and covariance structures, to an earlier multivariate estimator defined in \cite{abry2018waveleta,abry2018waveletb} and also to the univariate estimator defined in \cite{wendt2017multivariate}.
The potential benefits of the proposed eigen-wavelet estimation procedures are illustrated by means of epileptic seizure prediction from multichannel EEG data.
A documented toolbox  
for both the synthesis of multivariate selfsimilar processes and for the estimation of the vector of selfsimilarity parameters is publicly available at \url{github.com/charlesglucas/ofbm_tools}.

\section{Multivariate selfsimilarity: State of the art}
\label{sec.OFBM}

\subsection{Model: Operator Fractional Brownian Motion  (ofBm)}
\label{sec:def}

\subsubsection{Definition}
\label{sec:ofbmdef}

Operator Fractional Brownian Motion, originally proposed in \cite{mason:xiao:2002,Didier_G_2011_bernouilli_irpofbm,didier:pipiras:2011:exp}, is defined as a jointly Gaussian M-variate process
via the harmonizable representation
\begin{equation}
\label{equ:ofbmdef}
{\cal B}_{\HH,\mathbf{A}}(t) := \int_\real \frac{e^{itf}-1}{if} \left(f_+^{-  \HH + \frac{1}{2}\mathbb{I} } \mathbf{A} + f_-^{- \HH + \frac{1}{2}\mathbb{I}   } \overline{\mathbf{A}} \right) \tilde B(\mathrm{d}f).
\end{equation}
In \eqref{equ:ofbmdef}, $i^2=-1$, $ \mathbb{I} $ denotes the $ M \times M $ identity matrix,
 $\tilde B(\mathrm{d}f) $ a complex-valued Gaussian Hermitian random measure, such that  $ \E \tilde B(\mathrm{d}f) \tilde B(\mathrm{d}f)^\star = \mathrm{d}f$ (with $ ^\star $ the Hermitian transposition). Also,
$\mathbf{A}$ and $\HH$ are two $ M \times M $ complex- and real-valued matrices, respectively.
Furthermore, hereinafter it is also assumed that \cite{Didier_G_2011_bernouilli_irpofbm}: \\
\noindent {\sc Assumption OFBM1}: The Hurst matrix has the form
\begin{equation}
\HH = {\mathbf W} \diag(H_1, \hdots, H_M){\mathbf W}^{-1},
\end{equation}
with ${\mathbf W}$ a $ M \times M $ real-valued invertible matrix, and where
\begin{equation}\label{e:0<H1=<...=<HM<1}
0 < H_1 \leq \hdots \leq H_M < 1.
\end{equation}

\noindent {\sc Assumption OFBM2}: $\mathbf{A}\mathbf{A}^\star$ is positive definite, i.e., $\det \Re(\mathbf{A}\mathbf{A}^\star) > 0$. \\
\noindent {\sc Assumption OFBM3}: $\mathbf{A}\mathbf{A}^\star$ is a real-valued matrix.

OFBM2 ensures that  ${\cal B}_{\HH,\mathbf{A}}$  is a well-defined process~;
OFBM3 restricts ${\cal B}_{\HH,\mathbf{A}}$ to the special case of \emph{time-reversible} (or covariant under the change $t \rightarrow -t$) statistical properties~;
${\mathbf W}$ in OFBM1 is referred to as the mixing matrix and provides an important source of versatility of the model.

From Definition~\eqref{equ:ofbmdef}, the covariance matrix is obtained as:
\begin{equation}
\label{equ:ofbmcov}
\begin{aligned}
 \E \left[ {\cal B}_{\HH,\mathbf{A}}(t) {\cal B}_{\HH,\mathbf{A}}(s)^\star \right] = & \int_\real  \frac{(e^{itf}-1)(e^{-isf}-1)}{\vert if \vert^2} \\
& \left( f_+^{-\left( \HH-\frac{1}{2}\mathbb{I} \right) } \mathbf{A} \mathbf{A}^\star f_+^{- \left( \HH^\star-\frac{1}{2}\mathbb{I} \right) } \right.   \\
& \left. + f_-^{-\left( \HH-\frac{1}{2}\mathbb{I} \right) } \overline{\mathbf{A} \mathbf{A}^\star} f_-^{- \left( \HH^\star-\frac{1}{2}\mathbb{I} \right) }  \right)  \mathrm{d}f, \\
\end{aligned}
\end{equation}
where $\E$ stands for the ensemble average.

\subsubsection{Multivariate selfsimilarity}

Importantly, ofBm satisfies the general multivariate selfsimilarity relation \cite{Didier_G_2011_bernouilli_irpofbm}
\begin{equation}
\label{equ:ssmulti}
 \{{\cal B}_{\HH, \mathbf{A}}(t) \}_{t \in \real}  \stackrel{fdd}{=} \{a^{\HH} {\cal B}_{\HH, \mathbf{A}}(t/a)\}_{t \in \real}, \quad  \forall a > 0,
 \end{equation}
where $ a^{\HH} = \sum_{k \geq 0} \log^k a \hspace{0.5mm}\HH^k / k!$
and $\stackrel{fdd}{=} $ denotes equality for all finite-dimensional distributions.
In essence, Eq.~\eqref{equ:ssmulti} states that the statistics of ${ \cal B}_{\HH, \mathbf{A}}(t) $ are covariant under a change of time scale, $t \rightarrow t/a$, up to a change in amplitude ${\cal B} \rightarrow a^{\HH} {\cal B}$.
Importantly, this holds for any dilation factor $a>0$.
However, the factor controlling the change in amplitude $\rightarrow a^{\HH}$ involves jointly all entries of the vector  of selfsimilarity parameters $\underline{H} = (H_1,\ldots,H_M) $, through the mixing matrix ${\mathbf W}$.

Of particular pedagogical interest is the case when the scaling exponent matrix $\HH$ is diagonal, $\HH = \makebox{diag}(H_1,\ldots,H_M)$ (equivalently, ${\mathbf W} = \mathbb{I}$):
$M$-variate selfsimilarity in Eq.~\eqref{equ:ssmulti} then simplifies into a collection of $M$, component-wise, univariate selfsimilarity relations, each controlled by a single selfsimilarity parameter, i.e., $ \forall m= 1,\ldots,M $ and $ a > 0 $:
 \begin{equation}
 \label{equ:ss}
\{({\cal B}_{\HH, \mathbf{A}})_m(t) \}_{ t \in \real}  \stackrel{fdd}{=} \\
  \{a^{H_m}({\cal B}_{\HH, \mathbf{A}})_m(t/a)\}_{t \in \real}.
\end{equation}

\subsection{Eigen-wavelet analysis}

\subsubsection{Multivariate wavelet spectrum}
\label{sec:wavdef}

Let $\psi_0$ denote the so-called \emph{mother wavelet} defined as a compact support, unit norm reference pattern $\int_\real  \psi_0(t)^2 \, \mathrm{d}t = 1$, characterized by its number of vanishing moments $N_\psi$,  and satisfying regularity conditions quantified by the decay of its Fourier transform $\hat \psi_0$:
There is $\alpha >1$ such that $\underset{f \in \real}{\sup} \; \left\vert \hat \psi_0(f) \right\vert \, \left(1 + \vert f \vert \right)^\alpha < + \infty$.

Let $Y$ denote a $M \times N$ multivariate process, with $M$ and $N$ the number of components and the sample size.
The multivariate Discrete Wavelet Transform coefficients at scale $2^j$ consist of a $M \times N$ matrix $D_{Y}(j,k) = \int_\mathbb{R}  \psi(j,k) Y(t)\mathrm{d}t$, with  $\psi(j,k)=2^{-j/2} \psi_0(2^{-j}t-k)$ a collection of dilated and translated templates of $\psi_0$.
For further details on wavelet transforms, the reader is referred, e.g., to \cite{mallat:1999}.

The wavelet spectrum is defined as the collection of $M \times M$ scale-dependent random matrices:
\begin{equation}
\mathbf{S}_Y(2^j) =  \frac{1}{n_j}\sum^{n_j}_{k=1} D_Y(2^j,k) D_Y(2^j,k)^\top.
\end{equation}

\subsubsection{Multivariate wavelet spectrum for ofBm}

For $  {\cal B}_{\HH,\mathbf{A}}$, it has been shown that wavelet coefficients reproduce selfsimilarity as \cite{abry2018waveleta,abry2018waveletb}:
\begin{equation}
\label{eq:wavBH}
\left\{D_{ {\cal B}_{\HH,\mathbf{A}}}(2^j,k) \right\}_{k \in \ZZ}  \stackrel{fdd}{=}  \left\{2^{j(\HH +\frac{1}{2}\mathbb{I})} D_{ {\cal B}_{\HH,\mathbf{A}}}(2^0,k) \right\}_{k \in \ZZ}.
\end{equation}
The wavelet spectrum thus becomes:
\begin{equation}
\label{eq:wavBHA}
\mathbf{S}_{ {\cal B}_{\HH,\mathbf{A}}}(2^j) \stackrel{d}{=} 2^{j(\HH+\frac{1}{2}\mathbb{I})} \, S_{ {\cal B}_{\HH,\mathbf{A}}}(2^0) \, 2^{j(\HH^\top+\frac{1}{2}\mathbb{I})}.
\end{equation}
When ${\mathbf W} \neq \mathbb{I}$, Eq.~\eqref{eq:wavBHA} shows that (the expectation of) each entry of $\mathbf{S}_{ {\cal B}_{\HH,\mathbf A}}(2^j) $ generally consists of a mixture of power laws with respect to the scale $2^j$. In other words,
\begin{equation}
\label{eq:spectrum_mixing_power_law}
\mathbb{E} \mathbf{S}_{m,m'}(2^j) = \sum_{k=1}^M \sum_{k'=1}^M \alpha_{k,k'}^{m,m'} 2^{j(H_k+H_{k'}+1)},
\end{equation}
with $\alpha_{k,k'}^{m,m'}$ a priori depending on both $\HH$ and $\mathbf{A}$.

\subsubsection{Estimation of the vector of selfsimilarity parameters}

It has been shown in \cite{abry2018waveleta,abry2018waveletb} that each eigenvalue $\lambda_m(2^j)$ of $\E \mathbf{S}_{ {\cal B}_{\HH,\mathbf{A}}}(2^j)$ behaves asymptotically as a power law with respect to the scale $2^j$, where the scaling exponent is controlled by $H_m$,
\begin{equation}
\label{eq:eigen_power_law}
\lambda_m(2^j) \sim \xi_m(2^0) \, (2^j)^{2H_m+1}, \quad  \textrm{as } j \rightarrow +\infty.
\end{equation}
Eq.~\eqref{eq:eigen_power_law} naturally suggests that the estimation of the vector of selfsimilarity parameters $\underline{H}$ can be performed by a collection of individual,
for each $m = 1,\ldots, M$,  linear regressions on the logarithms of the eigenvalues $\hat \lambda_m(2^j)$ of $\mathbf{S}_{ {\cal B}_{\HH,\mathbf A}}(2^j) $ versus the logartihms of the scales $2^j$:
\begin{equation}
\label{eq:Mest}
\hat H_m^{\rm M} = \frac{1}{2} \left(\sum_{j=j_1}^{j_2} w_j \log_2 \hat \lambda_{m}(2^{j}) -1\right),
\end{equation}
where $ j \in (j_1, j_2)$ denote the range of octaves involved and the regression weights satisfy: $\sum_{j=j_1}^{j_2} j w_j =1 $ and $\sum_{j=j_1}^{j_2} w_j =0 $.
The asymptotic estimation performance was thoroughly studied theoretically in \cite{abry2018waveleta,abry2018waveletb}.

\section{Multivariate Fractional Brownian Motion}
\label{sec.MFBM}

\subsection{Definition}

For the specific \emph{nonmixing} case (${\mathbf W}=\mathbb{I}$), the selfsimilarity relation in Eq.~\eqref{equ:ss} suggests, as a first contribution of this work, the construction of another type of multivariate selfsimilar process. It is hereinafter referred to as multivariate fractional Brownian motion (mfBm) ${ {B}_{{\mathbf \Sigma},\mathbf{W},\Hv}}(t) $ and defined as a mixture of correlated univariate fractional Brownian motions:
\begin{equation}
\label{equ:mfBm}
 {B}_{{\mathbf \Sigma},{\mathbf W},\Hv}(t) := {\mathbf W} B_{{\mathbf \Sigma},\underline{H}}(t)^\top,
\end{equation}
where ${\mathbf W}$ denotes a $M \times M$ real-valued invertible matrix. Also, $B_{{\mathbf \Sigma},\underline{H}}(t) = [B_{{\mathbf \Sigma},H_1}(t),\ldots, B_{{\mathbf \Sigma},H_M}(t)]$ consists of a collection of $M$ fBms, $B_{{\mathbf \Sigma},H_m} $, $m=1,\ldots,M$, each with a possibly different selfsimilarity parameter $H_m$ ($0< H_m< 1$) and variance $\sigma^2_m$. In addition, these fBms are correlated via an $M \times M$ point matrix ${\mathbf \Sigma} = \diag(\sigma^2_1,\ldots,\sigma^2_M) \boldsymbol{\rho} \diag(\sigma^2_1,\ldots,\sigma^2_M)^\ast$, with $ \boldsymbol{\rho}$ the matrix of pairwise correlations.

From the definition of univariate fBm \cite{Samorodnitsky1994}, we can show that a mfBm $ {B}_{{\mathbf \Sigma},{\mathbf W},\Hv}$ is a zero-mean multivariate Gaussian process with covariance function given by:
\begin{equation}
\label{eq:cov_mfBm}
\begin{aligned}
\E \left[ B_{{\mathbf \Sigma},{\mathbf W},\Hv}(t) \right. & \left. B_{{\mathbf \Sigma},{\mathbf W},\Hv}(s)^\top \right] \\
& \begin{aligned}
=  \frac{1}{2} {\mathbf W} & \left(  \vert t \vert^{\mathrm{diag}(\Hv)}  \; {\mathbf \Sigma} \; \vert t \vert^{\mathrm{diag}(\Hv)} \right. \\
& \left. + \vert s \vert^{\mathrm{diag}(\Hv)}  \; {\mathbf \Sigma} \; \vert s \vert^{\mathrm{diag}(\Hv)} \right. \\
& \left.   - \vert t -s \vert^{\mathrm{diag}(\Hv)}  \; {\mathbf \Sigma} \; \vert t-s \vert^{\mathrm{diag}(\Hv)} \right) {\mathbf W}^\top.
\end{aligned}
\end{aligned}
\end{equation}
In particular, the covariance matrix of mfBm simply reads $\E \left[ B_{{\mathbf \Sigma},{\mathbf W},\Hv}(1) \right. \left. B_{{\mathbf \Sigma},{\mathbf W},\Hv}(1)^\top  \right]  = {\mathbf W} {\mathbf \Sigma} {\mathbf W}^\top$.

\subsection{mfBm vs. ofBm}

Comparing Eqs.\ \eqref{equ:ofbmcov} and \eqref{eq:cov_mfBm}, calculations not reported here yield the relations:
\begin{eqnarray}
\label{equ:ofbmnp}
{\mathbf A} {\mathbf A}^*  = {\mathbf W} \cdot ({\mathbf G} \odot {\mathbf \Sigma}) \cdot {\mathbf W}^\top \makebox{ and } \HH = {\mathbf W} \, \makebox{diag} (\underline{H}) \, {\mathbf W}^{-1}.
\end{eqnarray}
In \eqref{equ:ofbmnp}, ${\mathbf G}$ is a $M \times M$ matrix with entries $g_{m,m'} = \Gamma(H_{m}+H_{m'}+1) \sin((H_{m}+H_{m'})\pi/2) / (2 \pi) $, where $\Gamma$ is the standard Gamma-Euler function and the symbol $\odot $ stands for the Hadamard matrix product. Eq.~\eqref{equ:ofbmnp} shows that the mfBm ${B}_{{\mathbf \Sigma},{\mathbf W},\Hv}$ is a special case of the ofBm ${\cal B}_{\HH, {\mathbf A}}$, satisfying Assumptions OFBM1 to OFBM3.
The mfBm ${B}_{{\mathbf \Sigma},{\mathbf W},\Hv}$ thus provides a model for multivariate selfsimilarity. It can be considered a practical counterpart of the formal ofBm model, yet better suited for \emph{signal processing purposes} and applications to real-world data.
It splits the abstract matrices ${\mathbf A}$ and $\HH$ from the ofBm model into interpretable ingredients that are closer to applications: a vector of scaling exponents $\underline{H}$, a (premixing or intrinsic) covariance matrix ${\mathbf \Sigma}$, and a mixing matrix ${\mathbf W}$.
In particular, Eq.~\eqref{equ:ofbmnp} shows that mfBm ${B}_{{\mathbf \Sigma},{\mathbf W},\Hv}$ necessarily fulfills  Assumption OFBM3 and thus implies by construction time-reversible statistics.
An alternative permitting non time-reversible statistics had been considered in \cite{Amblard_P-0_2011_j-ieee-tsp_imfbm}, yet not allowing for mixing hence less versatile for applications.

Furthermore, the mfBm model allows us to better grasp a key feature of multivariate selfsimilarity. Starting with the simple bivariate setting, $M=2$, combining Assumption OFBM2 and Eq.~\eqref{equ:ofbmnp} yields:
\begin{equation}
\label{equ:vivcorr}
\boldsymbol{\rho}_{12}^2 \leq \rho_{max}^2 :=  \frac{\Gamma(2H_1+1)\Gamma(2H_2+1)\sin(\pi H_1)\sin(\pi H_2)}{(\Gamma(H_1+ H_2 +1)\sin(\pi/2 (H_1+H_2)))^2}.
\end{equation}
This shows that the correlation between components and differences in selfsimilarity parameters cannot be chosen independently.
This is illustrated in Fig.~\ref{fig:CorrConditionM2}: The largest allowed (squared) cross-correlation amongst components decreases when the difference $|H_2-H_1|$ increases, and conversely.
These calculations extend to $M$-variate settings with larger (absolute values of) correlations amongst components restricting the largest achievable discrepancies $|H_m - H_{m'}|$ amongst entries of the vector of selfsimilarity parameters (see also \cite{Amblard_P-0_2011_j-ieee-tsp_imfbm}).

\begin{figure}[!t]
\centerline{
\includegraphics[width=.6\linewidth]{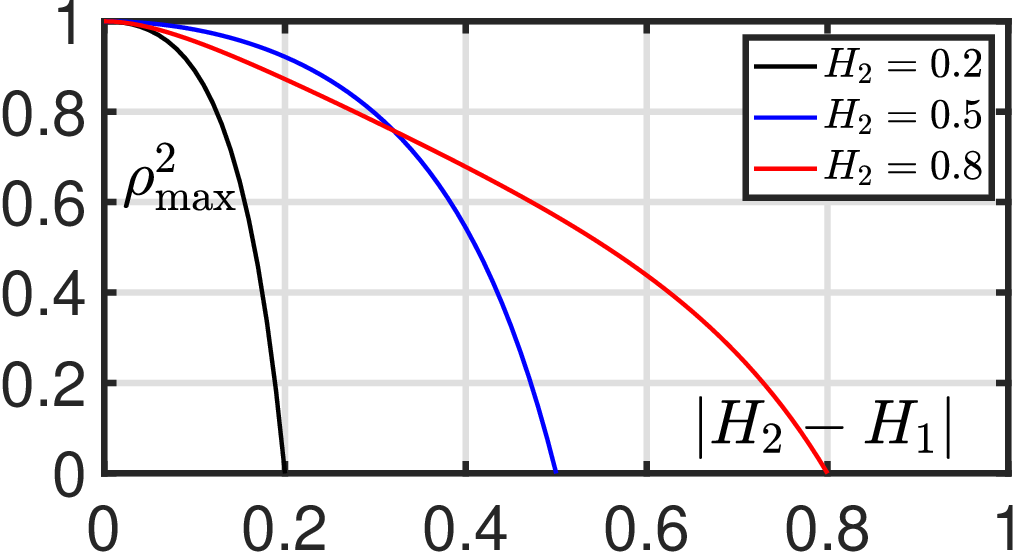}
}
\caption{\label{fig:CorrConditionM2}
{\bf Interplay between the largest possible difference in selfsimilarity parameters and the largest (squared) correlation coefficients,}  for different $H_2$ (bivariate setting).}
\end{figure}

Finally, $B_{{\mathbf \Sigma},\underline{H}}  \equiv {B}_{{\mathbf \Sigma},{\mathbf W}=\mathbb{I},\Hv}$ is obviously a special case of mfBm (cf., e.g., \cite{Coeurjolly_J-F_2013_ESAIM_wamfbm,Amblard_P-0_2011_j-ieee-tsp_imfbm,wendt2017multivariate}). Yet, it displays no mixing across components, rendering the model less versatile for practical purposes.

\subsection{Wavelet analysis of mfBm}

For mfBm, Eqs.~\eqref{eq:wavBH} and \eqref{eq:wavBHA} take the specific forms
\begin{equation}
\{D_{{B}_{{\mathbf \Sigma},{\mathbf W},\Hv}}(2^j,k)\}_{k \in \ZZ}  \stackrel{fdd}{=}  \left\{ {\mathbf W}  \, 2^{j(\underline{H} +\frac{1}{2})} \, D_{{B}_{{\mathbf \Sigma},\Hv}}(2^0,k) \right\}_{k \in \ZZ},
\end{equation}
\begin{equation}\label{eq:mfBm_spectrum}
\begin{aligned}
{\mathbf S}_{{B}_{{\mathbf \Sigma},{\mathbf W},\Hv}}(2^j) \stackrel{d}{=} {\mathbf W} \,  2^{j(\underline{H}+\frac{1}{2})}  \, {\mathbf S}_{{B}_{{\mathbf \Sigma},\Hv}}(2^0) \,  2^{j(\underline{H}+\frac{1}{2})} \, {\mathbf W}^\top ,
\end{aligned}
\end{equation}
with $2^{j(\underline{H}+\frac{1}{2})}= \makebox{diag}(2^{j(H_1+1/2)},\ldots,2^{j(H_M+1/2)})$.
This disentangles the contributions of the dilation vs.\ mixing operators to the shaping of the wavelet coefficients and spectrum.

Furthermore, Eqs.~\eqref{eq:spectrum_mixing_power_law} and \eqref{eq:eigen_power_law} obviously hold for mfBm. This permits to show, respectively, that both $\alpha_{k,k'}^{m,m'}$ and  $ \xi_m(2^0) $ depend on $\mathbf{W}$ and ${\mathbf \Sigma}$, but not on $\underline{H}$. An extension of the calculations in \cite{abry2018waveleta,abry2018waveletb} shows that $ \xi_m(2^0) $ actually depends only on the entries $\left({\mathbf S}_{{B}_{{\mathbf \Sigma},{\mathbf W},\Hv}}(2^0) \right)_{(m^{'},m^{''})}$  for $m\leq m', m^{''} \leq M$ (see Remark~\ref{remA}, Section~\ref{sec:appendixA}, supplementary material).

Finally and importantly, Eq.~\eqref{eq:eigen_power_law} demonstrates that the wavelet eigenvalue representations for both mfBm and ofBm disentangle the mixed power laws appearing in the wavelet spectrum ${\mathbf S}_{{B}_{{\mathbf \Sigma},{\mathbf W},\Hv}}(2^j)$. Thus, they restore the close relationship between power laws and scale-free dynamics, which had gotten altered by the mixing effect of the matrix ${\mathbf W}$. This motivates the estimation procedure $\underline{\hat H}^{\rm M}$ in Eq.~\eqref{eq:Mest}, which thus applies to mfBm as well as to ofBm.

\section{Selfsimilarity parameter vector estimation}
\label{sec.WOFBM}

\subsection{Repulsion bias-corrected multivariate estimation}
\subsubsection{Repulsion bias}

$\underline{\hat H}^{\rm M}$ in Eq.~\eqref{eq:Mest} is based on eigenvalues computed from the wavelet spectra ${\mathbf S}_{{\cal B}_{\HH, {\mathbf A}}}(2^j)$, namely, on the empirical wavelet coefficient covariance matrices estimated at different scales. However, it is well known that, for finite-size sample estimation of covariance matrices, eigenvalues undergo the so-called \emph{repulsion effect}:
Estimated eigenvalues tend to be farther apart one from the others, compared to true eigenvalues, all the more as the sample size is decreasing compared to the number of components  (cf. e.g., \cite{tao2012topics}).

In the wavelet spectrum representation, this \emph{repulsion effect} is enhanced by the fact that empirical wavelet coefficient covariance matrices are estimated with sample size $n_j$ that depends on analysis scales $ 2^j$, essentially as $n_j = 2^{-j}N$.
This induces scale-dependent repulsion biases in $\hat \lambda_m(2^j)$, and hence an overall bias in $ \underline{\hat H}^{\rm M}$  \cite{lucas2021bootstrap}.
The repulsion bias is illustrated in Fig.~\ref{fig:strucFunc_repulsEffect} (left), which displays $\log_2 \hat \lambda_m(2^j)$  vs $ \log_2 2^j = j$.
It is based on a synthetic mfBm with $M=6$ components of size $N=2^{12}$, with $H_1 =  \ldots =H_6 = 0.6$, a randomly selected orthonormal mixing matrix ${\mathbf W} \neq \mathbb{I}$ and a correlation matrix ${\mathbf \Sigma}$ with off-diagonal coefficients set to $0.8$. The five smallest curves should thus theoretically superimpose, whereas discrepancies of amplitudes in estimated eigenvalues, increasing with scales, can be observed.

\begin{figure}[!t]
\centerline{
\includegraphics[width=\linewidth]{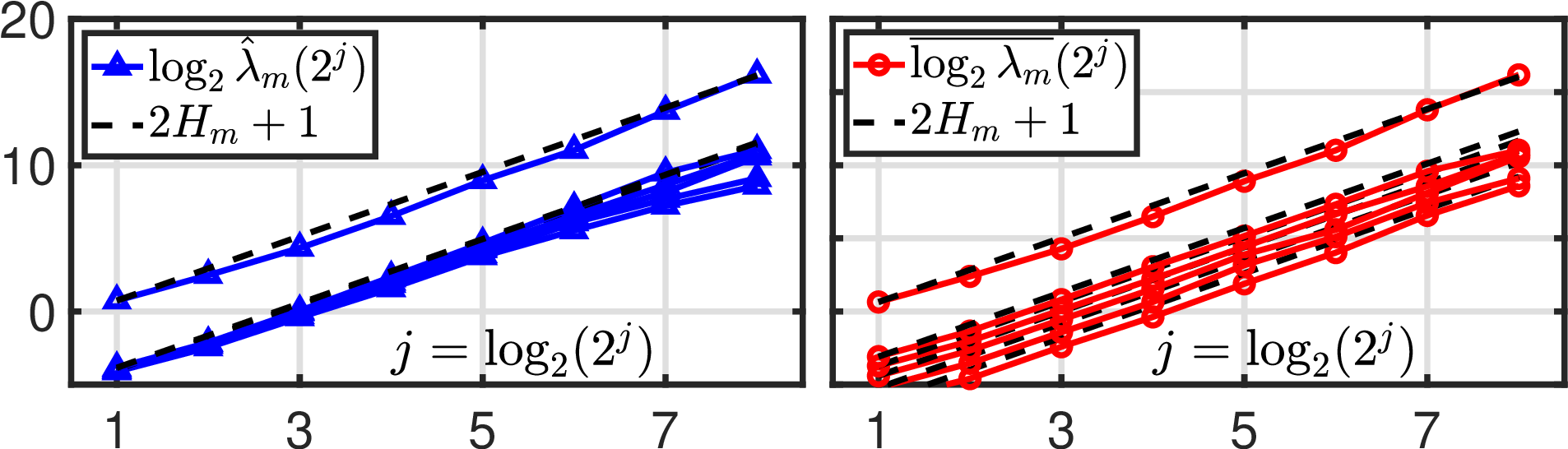}
}
\vspace{-.2cm}
\centerline{
\includegraphics[width=\linewidth]{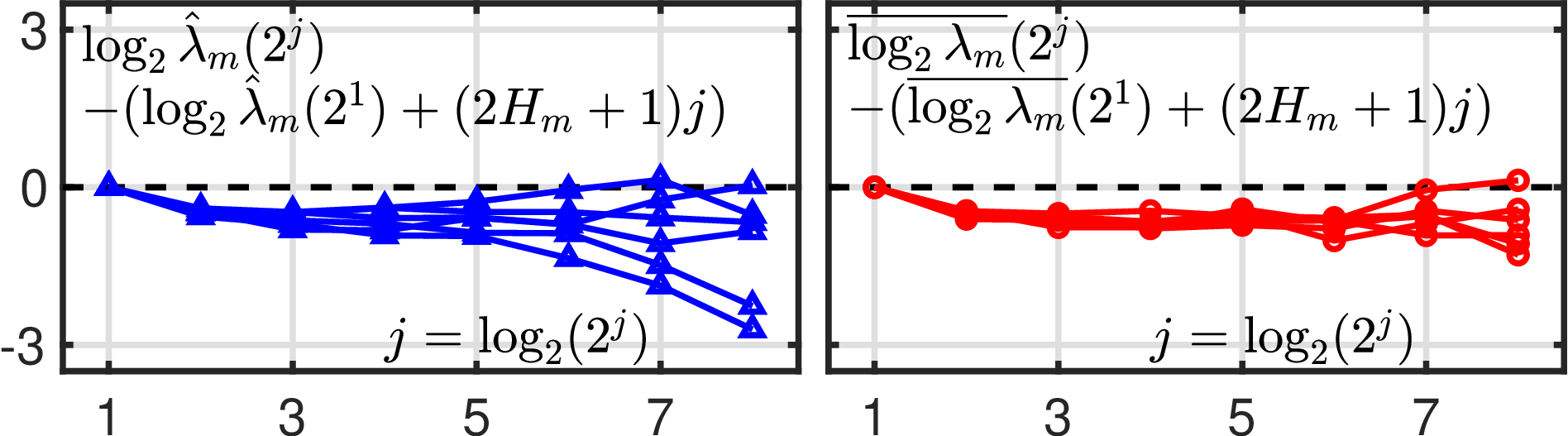}
}
\caption{\label{fig:strucFunc_repulsEffect}
{\bf Repulsion bias in eigenvalue estimation.} Top: Logarithms of estimated eigenvalues $\log_2 \hat \lambda_m(2^j)$ (left, blue) and $\overline{\log_2 \lambda_m}(2^j)$ (right, red), as functions of the logarithms of the scales $ \log_2 2^j = j$ for a synthetic mfBm with $M=6$, $N=2^{12}$, ${\mathbf \Sigma} \neq \mathbb{I}$, ${\mathbf W} \neq \mathbb{I}$ and $H_1 = \ldots = H_6 = 0.6$, and $j_2=8$. Superimposed straight lines indicate the theoretical slopes $ 2 H_m + 1$ (top, black). Bottom: Logarithms of estimated eigenvalues $\log_2 \hat \lambda_m(2^j)$ (left, blue) and $\overline{\log_2 \lambda_m}(2^j)$ (right, red) with theoretical linear behaviors subtracted.
}
\end{figure}

\subsubsection{Repulsion bias correction}

Elaborating on \cite{lucas2021bootstrap}, a novel estimation procedure for the estimation of $\underline{H}$ is devised and studied, hereinafter referred to as the \emph{bias-corrected} estimator $  \underline{\hat H}^{\rm M, bc} = (\hat H_1^{\rm M,bc}, \ldots, \hat H_M^{\rm M,bc})$. The main idea behind its construction lies in using the \textit{same number} of wavelet coefficients in the wavelet spectrum estimation at each scale. As a result, the eigenvalues computed from the wavelet spectra at different scales $2^j$ undergo an \emph{equivalent strength} repulsion effect, i.e., one that does not vary with the scale $2^j$. To that end, let $2^{j_2}$ denote the largest analysis scale used in estimation (cf. Eq.~\eqref{eq:Mest}) and let $n_{j_2}$ be the number of corresponding wavelet coefficients at that scale.
At each scale $2^{j_1} \leq 2^j \leq 2^{j_2}$, a collection of wavelet spectra are computed from $\tau=1,\ldots,2^{j_2-j}$
consecutive non-overlapping windows of wavelet coefficients with same size $n_{j_2}$:
\begin{equation}
{\mathbf S}^{(\tau)}(2^j) := \frac{1}{n_{j_2}}\sum^{\tau n_{j_2}}_{k=1+(\tau-1)n_{j_2}}D_{{\cal B}_{\HH,\mathbf A}}(2^j,k)D_{{\cal B}_{\HH,\mathbf A}}(2^j,k)^T.
\end{equation}
At each scale, the logarithms of the eigenvalues $\hat \lambda_m^{(\tau)}(2^j)$ computed from each wavelet spectrum ${\mathbf S}^{(\tau)}(2^j) $ are then averaged across windows:
\begin{equation}
\overline{\log_2\lambda_m}(2^j) := 2^{j-j_2}\sum_{\tau=1}^{2^{j_2-j}} \log_2 \hat \lambda_m^{(\tau)}(2^j ).
\end{equation}
The impact of this bias correction procedure on the scaling of the eigenvalues is illustrated in Fig.~\ref{fig:strucFunc_repulsEffect} (right),
which clearly shows that the theoretical slopes $2H_m+1$ are much better reproduced by the $\overline{\log_2\lambda_m}(2^j)$ than by the classical $\log_2 \hat \lambda_m(2^j)$. The bias-corrected multivariate estimates $ \hat H_m^{\rm M, bc}$ are thus defined as linear regressions across (the logarithms of the) scales $2^j$ of the averaged log-eigenvalues $\overline{\log_2\lambda_m}(2^j)$:
\begin{equation}
 \label{eq:bcHurstEstimation}
 \hat H_m^{\rm M, bc} = \frac{1}{2}\left(\sum_{j=j_1}^{j_2} w_j \overline{\log_2 \lambda_m}(2^{j}) - 1 \right).
 \end{equation}
 The asymptotic estimation performance is studied theoretically in Section~\ref{sec:theoreticalPerf}. The finite-size estimation performance is investigated numerically in Section~\ref{sec.MC}.

\subsubsection{Nonmixing case}

In the nonmixing case, i.e., when ${\mathbf W}=\mathbb{I}$, the selfsimilarity matrix reduces to $\HH=\mathrm{diag}(\Hv)$.
While the multivariate estimators $ \hat H_m^{\rm M} $  in Eq.~\eqref{eq:Mest} and  $\hat H_m^{\rm M, bc}$ in Eq.~\eqref{eq:bcHurstEstimation} obviously apply to the nonmixing case, a natural estimation procedure, extending univariate analysis, consists in performing linear regressions on the logarithms of the diagonal coefficients $ {\mathbf S}_{m,m}(2^{j})$ only:
\begin{equation}
\label{eq:UEst}
\hat H_m^{\rm U} = \frac{1}{2} \left(\sum_{j=j_1}^{j_2} w_j \log_2 {\mathbf S}_{m,m}(2^{j}) -1\right).
\end{equation}
The resulting estimators $\hat H_m^{\rm U}$, intrinsically \emph{univariate} in spirit, were intensively studied in \cite{wendt2017multivariate} for nonmixing mfBm.

\subsection{Theoretical study of asymptotic estimation performance}
\label{sec:theoreticalPerf}

\subsubsection{Asymptotic framework}
\label{sec:asymtpf}
The estimation performance is studied theoretically in the asymptotic double and joint limits of large sample sizes $N \rightarrow +\infty$, with linear regressions performed at coarse scales, in the range $(j_1(N),j_2(N)) \rightarrow +\infty$.
Technically, the scaling range is defined as $(j_1(N) = j_1^0 + \log_2 a(N),j_2(N) = j_2^0 + \log_2 a(N))$. The pair $(j_1^0,j_2^0)$ is an arbitrarily chosen range for $N_0$, and $a(N)$ must satisfy:
\begin{equation}\label{e:a(N)}
a(N)2^{j_2} \leq N,\quad
\frac{a(N)}{N}+ \frac{N}{a(N)^{2 \varpi +1}} \longrightarrow 0, \textrm{ as } N \rightarrow \infty,
\end{equation}
\begin{equation}
\label{e:varpi_parameter}
\varpi = \min \left\{\min_{ \{ 1 \leq i \leq M \vert \; H_{i}-H_{i-1}>0 \} }(H_{i}-H_{i-1}), \frac{H_{1}}{2} + \frac{1}{4} \right\}
\end{equation}
(where, for notational simplicity, $H_0 := 0$).
These conditions essentially imply that $\log_2 a(N)$ increases as $\beta \log_2 N$, with $ 1/(2\varpi +1) < \beta < 1$. They also imply that the number of scales involved in linear regressions does not vary with sample size $N$:
$j_2(N)-j_1(N)+1 = j_2^0-j_1^0+1$.

\subsubsection{Asymptotic consistency}

Theorem \ref{t:H-hat_m_is_consistent} establishes the asymptotic consistency of the bias-corrected estimator $\underline{\hat H}^{\rm M,bc} $:
\begin{theorem}\label{t:H-hat_m_is_consistent}
Under Assumptions (OFBM1$-$3) and for a mother wavelet as defined in Section~\ref{sec:wavdef}, (with $\stackrel{\bbP}\rightarrow$ the convergence in probability):
\begin{enumerate}
\item For any $j = j_1^0,\hdots,j_2^0$, $\tau  = 1,\hdots,2^{j_2^0-j}$ and $m = 1,\hdots, M$, as $N \rightarrow \infty$,
\begin{equation}\label{e:lim_n_a*lambda/a^(2h+1)}
\frac{\hat \lambda_{m}^{(\tau)}(a(N)2^j)}{a(N)^{2H_m+1}} \stackrel{\bbP}\rightarrow \xi_{m}(2^j) = (2^j)^{2H_m+1}\xi_{m}(2^0)  > 0,
\end{equation}
where $\hat \lambda^{(\tau)}_{1}(a(N)2^j), \ldots,\hat \lambda^{(\tau)}_{M}(a(N)2^j)$ are the eigenvalues of ${\mathbf S}^{(\tau)}(a(N)2^j)$.
Likewise, as $N \rightarrow \infty$,
\begin{equation}\label{e:lim_n_a*lambda(EW)/a^(2h+1)}
\frac{\lambda^{(\tau)}_{m}(a(N)2^j)}{a(N)^{2H_{m}+1}} \rightarrow \xi_{m}(2^j)=(2^j)^{2H_m+1}\xi_{m}(2^0)  > 0,
\end{equation}
where $\lambda^{(\tau)}_{1}(a(N)2^j), \ldots, \lambda^{(\tau)}_{M}(a(N)2^j)$ are the eigenvalues of $\bbE {\mathbf S}^{(\tau)}(a(N)2^j)$.
\item Let $a(N)$ be as in \eqref{e:a(N)}. Then, as $N \rightarrow \infty$,
\begin{equation}\label{e:H-hat_m_is_consistency}
\hat H^{\rm M,bc}_{m} \stackrel{\bbP}\rightarrow H_m, \quad \forall m = 1,\hdots, M,
\end{equation}
\end{enumerate}
\end{theorem}

\begin{proof}
For ({\it i}), the proof follows from adapting the technique for establishing Theorem 3.1 in \cite{abry:boniece:didier:wendt:2022:prob}, see Section~\ref{sec:appendixA} of the supplementary material for details.

For ({\it ii}), Eq.~\eqref{e:lim_n_a*lambda/a^(2h+1)} leads to the following equivalence:
\begin{align}
\hat H^{\rm M,bc}_{m}  
& =\frac{1}{2} \sum^{j_2^0}_{j=j_1^0} \frac{w_j }{2^{j_2^0-j}} \sum^{2^{j_2^0-j}}_{\tau=1}  \Big( \log_2 \frac{\hat \lambda^{(\tau)}_{m}(a(N)2^j)}{a(N)^{2H_m+1}}  \\ \nonumber
& \qquad \qquad - \log_2 a(N)^{2H_m+1} \Big) - \frac{1}{2} \\ \nonumber
& \stackrel{\bbP}\sim \frac{1}{2} \sum^{j_2^0}_{j=j_1^0} \frac{w_j }{2^{j_2^0-j}} \sum^{2^{j_2^0-j}}_{\tau=1}  \Big( \log_2 (2^j)^{2H_m+1}\xi_{m}(2^0)  \\ \nonumber
& \qquad \qquad - \log_2 a(N)^{2H_m+1} \Big) - \frac{1}{2} \\ \nonumber
& \stackrel{\bbP}\sim \frac{1}{2} \sum^{j_2^0}_{j=j_1^0} w_j \Big( (2H_m+1)j+ \log_2 \frac{\xi_{m}(2^0)}{a(N)^{2H_m+1}}  \Big) - \frac{1}{2}.
\end{align}
Hence, by the definition of $w_j$, as $N \rightarrow \infty$,
\begin{equation}
\hat H^{\rm M,bc}_{m} \stackrel{\bbP}\sim \frac{1}{2} [(2H_m + 1)] - \frac{1}{2} = H_m.
\end{equation}
This establishes \eqref{e:H-hat_m_is_consistency}.
\end{proof}

\subsubsection{Simple eigenvalues condition}
\label{sec:C0}
Our theoretical study of the asymptotic distribution relies on the condition that, asymptotically, all wavelet spectra eigenvalues are simple. \\
\noindent {\sc Condition (C0)}:
 $\forall m_1, m_2 \in \{1,\ldots, M\}$, with $m_1 \neq m_2$,
\begin{equation}\label{e:H_ell1_neq_H_ell2_=>_xi_ell1_neq_xi_ell2}
H_{m_1} = H_{m_2}  \Rightarrow \xi_{m_1}(2^0) \neq \xi_{m_2}(2^0),
\end{equation}
which states that either the selfsimilarity parameters $H_m$ are different, or when they are equal, $H_{m_1} = H_{m_2}$, $m_1 \neq m_2$, the constants $\xi_{m}(2^0)$ in Eq.~\eqref{eq:eigen_power_law} are different.

\subsubsection{Asymptotic normality}

Theorem \ref{t:scale_invariance_of_the_distribution} establishes the asymptotic normality  of the bias-corrected estimator $\underline{\hat H}^{\rm M,bc} $:
\begin{theorem}\label{t:scale_invariance_of_the_distribution}
Let $n_{a,j} = \frac{N}{a(N)2^j}$, $ j = j^{0}_1,\hdots,j^{0}_2,$
be the number of wavelet coefficients available at scale $a(N)2^j$.
Under Assumptions (OFBM1$-$3), for a mother wavelet as defined in Section~\ref{sec:wavdef}, and assuming that (C0) holds,
with $a(N)$ defined in \eqref{e:a(N)} and where $\stackrel{d}\rightarrow$ denotes the convergence in distribution:
\begin{enumerate}
\item Then, as $N \rightarrow \infty$,
\begin{align}\label{e:scale_invariance_of_the_distribution}
\left\{\sqrt{n_{a,j_2^0}}\big( \log_2 \right. & \left. \hat \lambda_{m}^{(\tau)}(a(N)2^j) \right. & \\ \nonumber
&\left. - \log_2 \right.  \left. \lambda_{m}^{(\tau)}(a(N)2^j) \big) \right\}^{j=j_1^0,\hdots,j_2^0}_{\substack{\vspace{3mm} \\ m=1,\hdots,M \hfill \\ \tau= 1,\hdots,2^{(j_2^0-j_1^0)-j}} }  & \\ \nonumber
   & \stackrel{d}\rightarrow {\mathcal N}(0,{\mathbf \Upsilon}_{\bf B}),
 \end{align}
where ${\mathbf \Upsilon}_{\bf B}$ is a symmetric positive semidefinite matrix.
\item Moreover, as $N \rightarrow \infty$,
\begin{equation}\label{e:H-hat_m_is_asymptotically_normal}
\Big\{\sqrt{N/a(N)}\hspace{0.5mm}\big(\hat H^{\rm M,bc}_{m} - H_m\big)\Big\}_{m = 1,\hdots, M} \stackrel{d}\rightarrow {\mathcal N}(0,{\mathbf \Sigma}_{\bf B}),
\end{equation}
for some ${\mathbf \Sigma}_{\bf B} \in {\mathcal S}_{\geq 0}(M,\bbR)$.
\end{enumerate}
\end{theorem}
\begin{proof}
The proof follows from adapting the proof of asymptotic normality for the non bias-corrected multivariate estimator $\underline{\hat H}^{\rm M} $ from Theorem 3.2 in \cite{abry:boniece:didier:wendt:2022:prob}, as detailed in Section~\ref{sec:appendixB} of the supplementary material.
\end{proof}

\section{Empirical finite-size estimation performance}
\label{sec.MC}

The goal of the present section is to complement the theoretical study of estimation performance in the limit of large scales and large sample sizes. We do so by investigating their finite sample size counterparts by means of Monte Carlo simulations, conducted on synthetic mfBm. Different parametric configurations $({\mathbf \Sigma}, {\mathbf W}, \underline{H})$ are used so as to further assess their impact on estimation performance.

The performance of the proposed bias-corrected multivariate estimator $\underline{\hat H}^{\rm M,bc} $ (Eq.~\eqref{eq:bcHurstEstimation}) is compared, in terms of biases, variances, Mean Squared Error (MSE), and covariance structures, to those of the multivariate estimator $\underline{\hat H}^{\rm M} $ (Eq.~\eqref{eq:Mest}) and of the univariate estimator $\underline{\hat H}^{\rm U} $ (Eq.~\eqref{eq:UEst}).

\subsection{Numerical simulation set-up}

The Monte Carlo simulations involve $N_{\rm MC}=1000$ independent realizations of synthetic mfBm.
The increments $\{ B_{{\mathbf \Sigma},{\mathbf W},\Hv}(t+1) - B_{{\mathbf \Sigma},{\mathbf W},\Hv}(t) \}_{t =1,\ldots,N }$ of a mfBm, referred to as multivariate fractional Gaussian noise, form a multivariate stationary Gaussian process with covariance directly stemming from Eq.~\eqref{eq:cov_mfBm}. Thus, it can be synthesized numerically by circulant embedding, cf. e.g.,~\cite{Helgason_H_2011_j-sp_fessmgtsuce}.

Results reported here are restricted to $M=6$ for the sake of clarity~;
additional results for $M=20$, available upon request, led to identical conclusions. The range of sample sizes covered is $N \in \{2^{13},\ldots,2^{18} \}$.
The parameters ${\mathbf \Sigma}$, ${\mathbf W}$ and $\Hv$ are chosen to investigate a set of representative configurations: \\
\noindent {\bf Config1} (${\mathbf \Sigma} \neq \mathbb{I}$, ${\mathbf W} \neq \mathbb{I}$, $H_m$ all different) 
Generic situation involving mixing, nonzero correlations amongst components, as well as different selfsimilarity parameters. This type of situation calls for multivariate estimators ($\underline{\hat H}^{\rm M,bc} $, $\underline{\hat H}^{\rm M} $ )~; \\
\noindent {\bf Config2} (${\mathbf \Sigma} \neq \mathbb{I}$, ${\mathbf W} \neq \mathbb{I}$, $H_m$ all equal) 
Generic situation with mixing and correlations amongst components. Yet, all selfsimilarity parameters are assumed equal. Thus, univariate estimation $\underline{\hat H}^{\rm U} $ could be used in this multivariate setting, as much as the $\underline{\hat H}^{\rm M,bc} $, $\underline{\hat H}^{\rm M} $~; \\
\noindent {\bf Config3} (${\mathbf \Sigma} \neq \mathbb{I}$, ${\mathbf W} =\mathbb{I}$, $H_m$ all different)
Nonmixing situation, ence, univariate  $\underline{\hat H}^{\rm U} $ and multivariate estimation ($\underline{\hat H}^{\rm M,bc}, \underline{\hat H}^{\rm M}$) are all legitimate procedures~; \\
\noindent {\bf Config4} (${\mathbf \Sigma} \neq \mathbb{I}$, ${\mathbf W} =\mathbb{I}$, $H_m$ all equal)
Nonmixing situations, yet, all selfsimilarity parameters are equal, and thus both univariate $\underline{\hat H}^{\rm U} $ and multivariate estimations ($\underline{\hat H}^{\rm M,bc}, \underline{\hat H}^{\rm M}$) should be appropriate procedures.\\
\noindent {\bf Config5} (${\mathbf \Sigma} = \mathbb{I}$, ${\mathbf W} \neq \mathbb{I}$, $H_m$ all different) 
Mixing situations with different selfsimilarity parameters, yet no correlation. 

A generic covariance matrix ${\mathbf \Sigma}$ is chosen: ${\mathbf \Sigma}_{m,m'}=0.7^{\vert m-m' \vert}$, $1 \leq m \leq m' \leq M$.
A mixing matrix ${\mathbf W}$ is randomly drawn from the set of  $M \times M$ real-valued non-orthogonal matrices and kept fixed for all experiments.
The selfsimilarity parameter vector $\Hv=(H_1,\ldots,H_6)$ is either set to $H_m = 0.7$, $\forall m$, or $\Hv = (0.4,0.5,0.6,0.65,0.7,0.8)$.
These different
configurations are chosen so that {\sc Condition (C0)} holds -- this condition is required for asymptotic normality claims (cf. Section~\ref{sec:C0}).
Wavelet analysis is performed using the least asymmetric Daubechies2 mother wavelet~\cite{mallat:1999}.
Identical conclusions can be obtained using other  wavelets satisfying $N_\psi \geq 2$.
Linear regressions are performed, with $w_j$ as defined in \cite{wendt2017multivariate}, over scales that depend on $N$, as analyzed theoretically in Section~\ref{sec:asymtpf} (Eq.~\eqref{e:a(N)}): 
$2^{j_1}=a(N)2^6$ to $2^{j_2}=a(N)2^{9}$, where $a(N) = 2^{\lfloor \beta \log_2 (N/N_0) \rfloor}$ with $ \beta=0.9$ and $N_0=2^{13}$.

\subsection{Asymptotic normality assessment}

The normality of $\underline{\hat H}^{\rm M,bc}$ is assessed by computing quantile-quantile plots for the squared Mahalanobis distance,
$$(\underline{\hat H}^{\rm M,bc} - \bbE \underline{\hat H}^{\rm M,bc}) \mathrm{Var}(\underline{\hat H}^{\rm M,bc})^{-1} \; (\underline{\hat H}^{\rm M,bc} - \bbE \underline{\hat H}^{\rm M,bc})^\top,$$
where the ensemble average $\E$ is
computed as the mean across independent realizations.
Samples of this Mahalanobis distance
 are plotted against a $\chi^2$ distribution with $M$ degrees of freedom, since this distribution is expected to hold under the exact joint normality of $\underline{\hat H}^{\rm M,bc}$. Fig.~\ref{fig:asymptoticNormalityN} reports such quantile-quantile plots, for the different configurations and for different sample sizes. It shows very satisfactory straight lines, even at the smallest sample size $N=2^{13}$. This suggests that the joint normality for $\underline{\hat H}^{\rm M,bc}$ holds very satisfactorily for finite sample size estimation, even for small sample sizes. It thus indicates fast convergence to asymptotic normality, proven theoretically in Theorem~\ref{t:scale_invariance_of_the_distribution} (Eq.~\eqref{e:H-hat_m_is_asymptotically_normal}).

\begin{figure}[!t]
\centerline{\includegraphics[width=\linewidth]{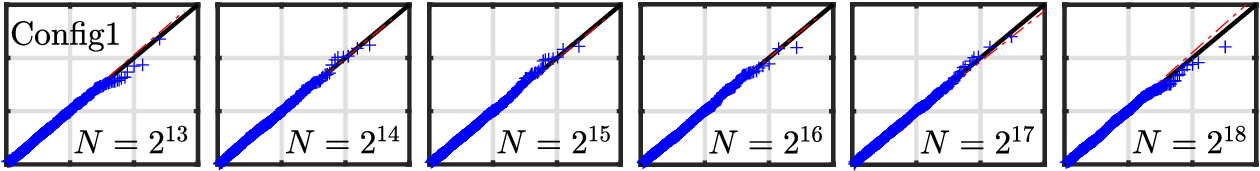}}
\centerline{\includegraphics[width=\linewidth]{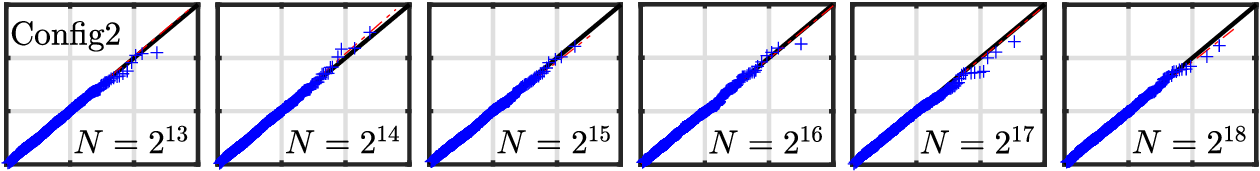}}
\centerline{\includegraphics[width=\linewidth]{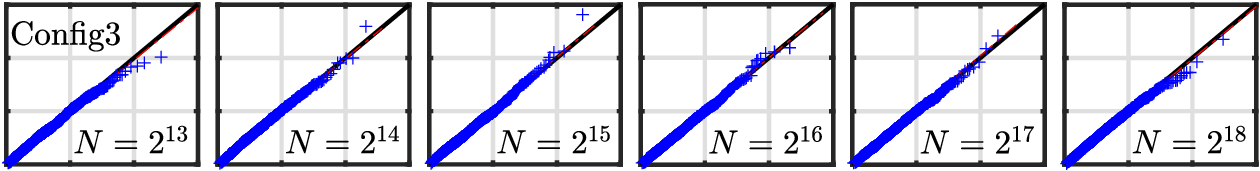}}
\centerline{\includegraphics[width=\linewidth]{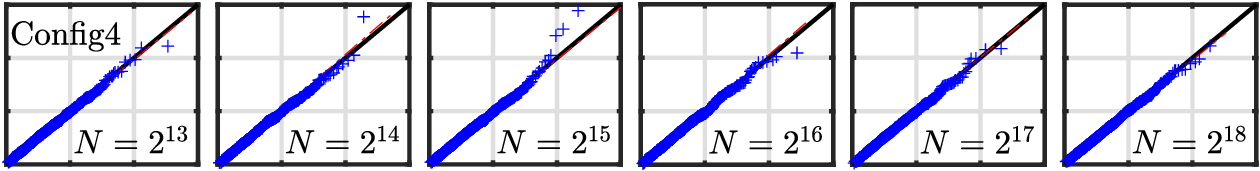}}
\centerline{\includegraphics[width=\linewidth]{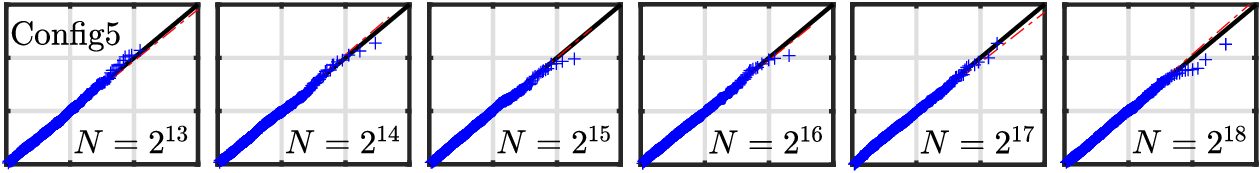}}
\vspace{-.1cm}
\caption{\label{fig:asymptoticNormalityN} {\bf Finite-size joint normality for $\boldsymbol{\underline{\hat H}^{\rm M,bc}}$.} Quantile-quantile plots of the empirical distribution of the Mahalanobis distance computed from the bias-corrected multivariate estimates $\underline{\hat H}^{\rm M,bc}$ against a $\chi^2$ distribution with $M$ degrees of freedom (expected theoretically under joint normality), for different sample sizes $N$ and different configurations.}
\end{figure}

\subsection{Estimation performance: bias, covariance, MSE}

For any estimate of the vector of selfsimilarity parameters $\underline{\hat H}$, matrices of bias, covariance and MSE are defined as
\begin{align}
\textrm{Bias}^2(\underline{\hat H}) & =  ( \mathbb{E}\underline{\hat H} - \underline{H} )(\mathbb{E}\underline{\hat H} - \underline{H} )^\top, \\
\textrm{Cov}(\underline{\hat H}) & =  \mathbb{E} \left[ (\mathbb{E}\underline{\hat H} -\underline{\hat H} ) (\mathbb{E}\underline{\hat H} -\underline{\hat H} )^\top \right], \\
\mathrm{MSE}(\underline{\hat H}) & = \mathbb{E} \left[ (\underline{\hat H}- \underline{H} ) (\underline{\hat H}- \underline{H} )^\top \right].
\end{align}
The estimation performance is quantified by the spectral norm of these matrices, i.e., by the largest absolute value of their eigenvalues ~\cite{meyer2023matrix}.
\begin{figure}[!t]
\centering
\includegraphics[width=\linewidth]{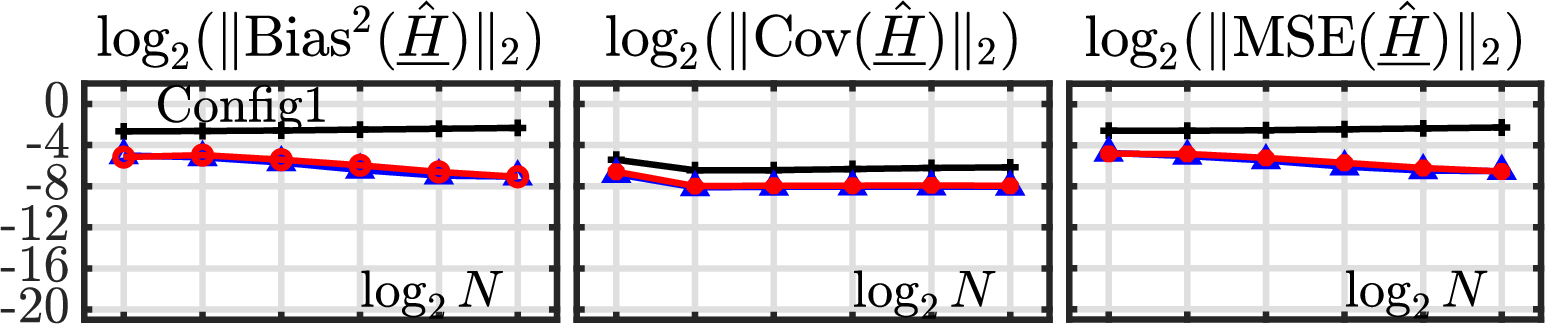}
\includegraphics[width=\linewidth]{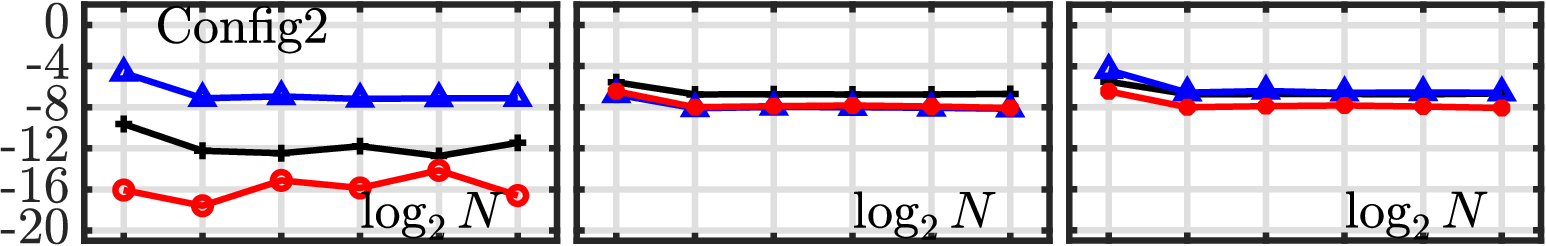}
\includegraphics[width=\linewidth]{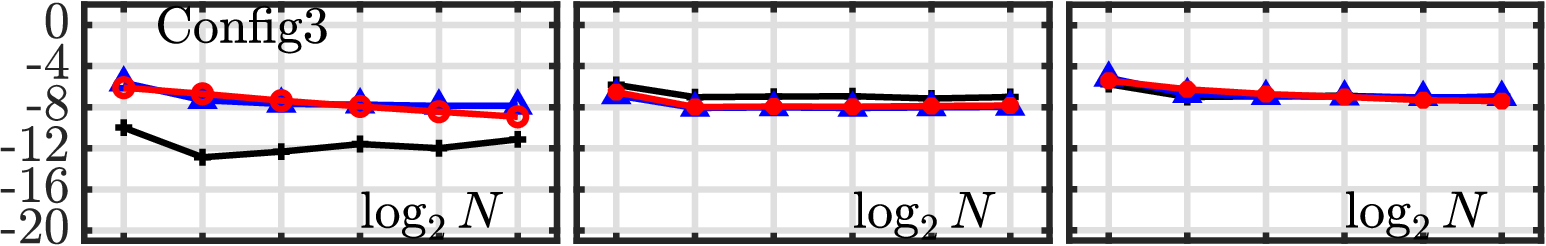}
\includegraphics[width=\linewidth]{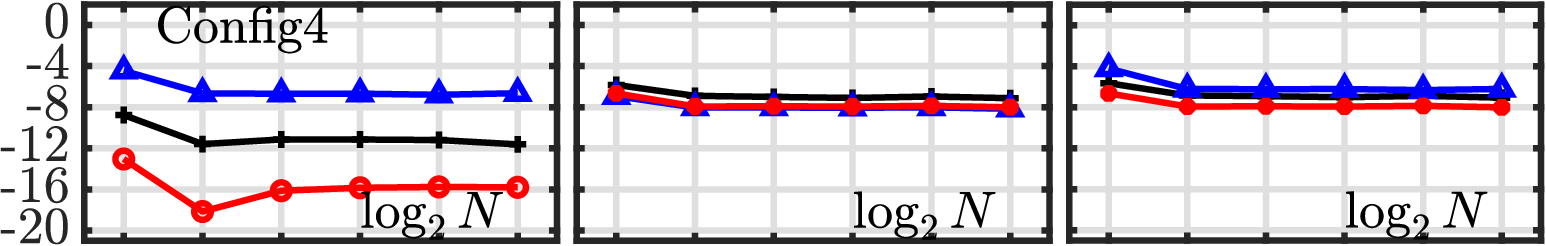}
\includegraphics[width=\linewidth]{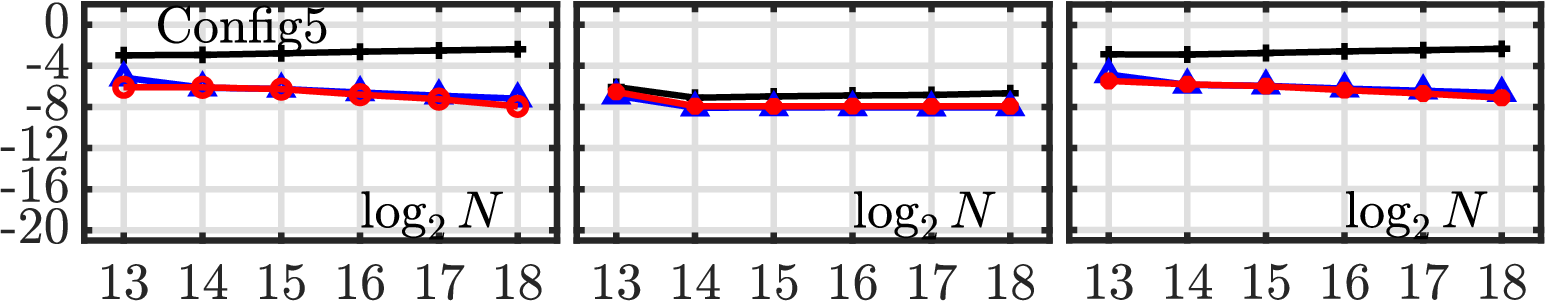}
\caption{\label{fig:figperf}
{\bf Estimation performance.} {Logarithm of the spectral norms of the bias (left), covariance (middle) and MSE (right) for  the proposed bias-corrected multivariate estimator $\underline{\hat H}^{\rm M,bc}$ (red), compared to those for the classical multivariate estimator $\underline{\hat H}^{\rm M}$ (blue) and for the univariate estimator $\underline{\hat H}^{\rm U}$ (black), as functions of the logarithm of the sample size $N$.}
}
\end{figure}

Fig.~\ref{fig:figperf} reports the (logarithms of the) spectral norms of the bias, covariance and MSE for the proposed bias-corrected multivariate estimator $\underline{\hat H}^{\rm M,bc}$, compared to those for the classical multivariate estimator $\underline{\hat H}^{\rm M}$ and for the univariate estimator $\underline{\hat H}^{\rm U}$. The results are shown as functions of the logarithm of the sample size $N$ and for each of the different configurations.

For the generic case represented by {\bf Config1} (Fig.~\ref{fig:figperf}, first row), the univariate estimator $\underline{\hat H}^{\rm U}$ suffers from significant bias stemming from the component mixing caused by ${\mathbf W}$. Thus, it is significantly outperformed by the multivariate estimators $\underline{\hat H}^{\rm M}$ and $\underline{\hat H}^{\rm M,bc}$. The latter show equivalent performance as the combination of  ${\mathbf \Sigma}$, ${\mathbf W}$ and different $H_m$ yield different eigenvalues for $\E {\mathbf S}(2^j)$ at all scales and thus small \emph{repulsion bias} in their estimated values.

For {\bf Config2} (Fig.~\ref{fig:figperf}, second row), with mixing yet equal $H_m$, the eigenvalues of $\E {\mathbf S}(2^j)$ lie much closer to each other at most scales and hence suffer from substantial repulsion biases in their estimation.
This results in a significant bias for the classical multivariate estimator $\underline{\hat H}^{\rm M}$.
Fig.~\ref{fig:figperf}, middle row, shows that the proposed bias-corrected multivariate estimator $\underline{\hat H}^{\rm M,bc}$ succeeds in significantly reducing bias without an increase in the variance. This leads to an overall much improved estimation performance as quantified by the MSE.
While the univariate estimator $\underline{\hat H}^{\rm U}$ displays small bias, as expected in view of the equality of all selfsimilarity parameters, it also shows greater variance, thus, greater MSE, when compared to the proposed $\underline{\hat H}^{\rm M,bc}$.

For the nonmixing case, {\bf Config3} (Fig.~\ref{fig:figperf}, third row), the univariate estimator $\underline{\hat H}^{\rm U}$ naturally shows much smaller biases compared to the multivariate estimators $\underline{\hat H}^{\rm M}$ and $\underline{\hat H}^{\rm M,bc}$. Yet, it also displays larger variances. Thus, as a result of bias-variance trade-off, the MSE of all three estimators are equivalent.
Importantly, this shows that even for data corresponding to nonmixing situations, there is no cost in estimation performance associated with the use of multivariate estimators.

For the nonmixing case with equal $H_m$, {\bf Config4} (Fig.~\ref{fig:figperf}, fourth row), the univariate estimator $\underline{\hat H}^{\rm U}$ and bias-corrected multivariate estimator $\underline{\hat H}^{\rm M,bc}$ perform similarly to {\bf Config2}. This suggests that, for $H_m$ all equal, the presence of mixing has no impact in terms of bias or variance. However, estimation performance of the multivariate estimator $\underline{\hat H}^{\rm M}$ is slightly more affected than in {\bf Config2}, showing that the repulsion bias is more substantial in the absence of mixing.

Finally, for the mixing case with different $H_m$ and no correlation, {\bf Config5} (Fig.~\ref{fig:figperf}, fifth row), the performances are similar to those for {\bf Config1} for all three estimators, indicating their robustness against component correlations ${\mathbf \Sigma}$.

All together, the results reported in Fig.~\ref{fig:figperf} lead us to conclude that the proposed bias-corrected multivariate estimator $\underline{\hat H}^{\rm M,bc}$ displays the best estimation performance regardless of the actual configuration of the parameters.
In practice it is not a priori known whether or not mixing is present, or whether selfsimilarity parameters are pairwise distinct or equal. We have seen that the estimator $\underline{\hat H}^{\rm U}$ fails when mixing is present. Also, we have found out that the estimator $\underline{\hat H}^{\rm M}$ fails when repulsion bias is present, i.e., when one needs to estimate identical selfsimilarity parameters. By contrast, $\underline{\hat H}^{\rm M,bc}$ displays neither one of these limitations -- this versatility is clearly an important property that motivates its practical use.

\subsection{Covariance structure of $\underline{\hat H}^{\rm M,bc}$}

\subsubsection{Correlation structure}

\begin{figure}[!t]
\centering
\includegraphics[width=\linewidth]{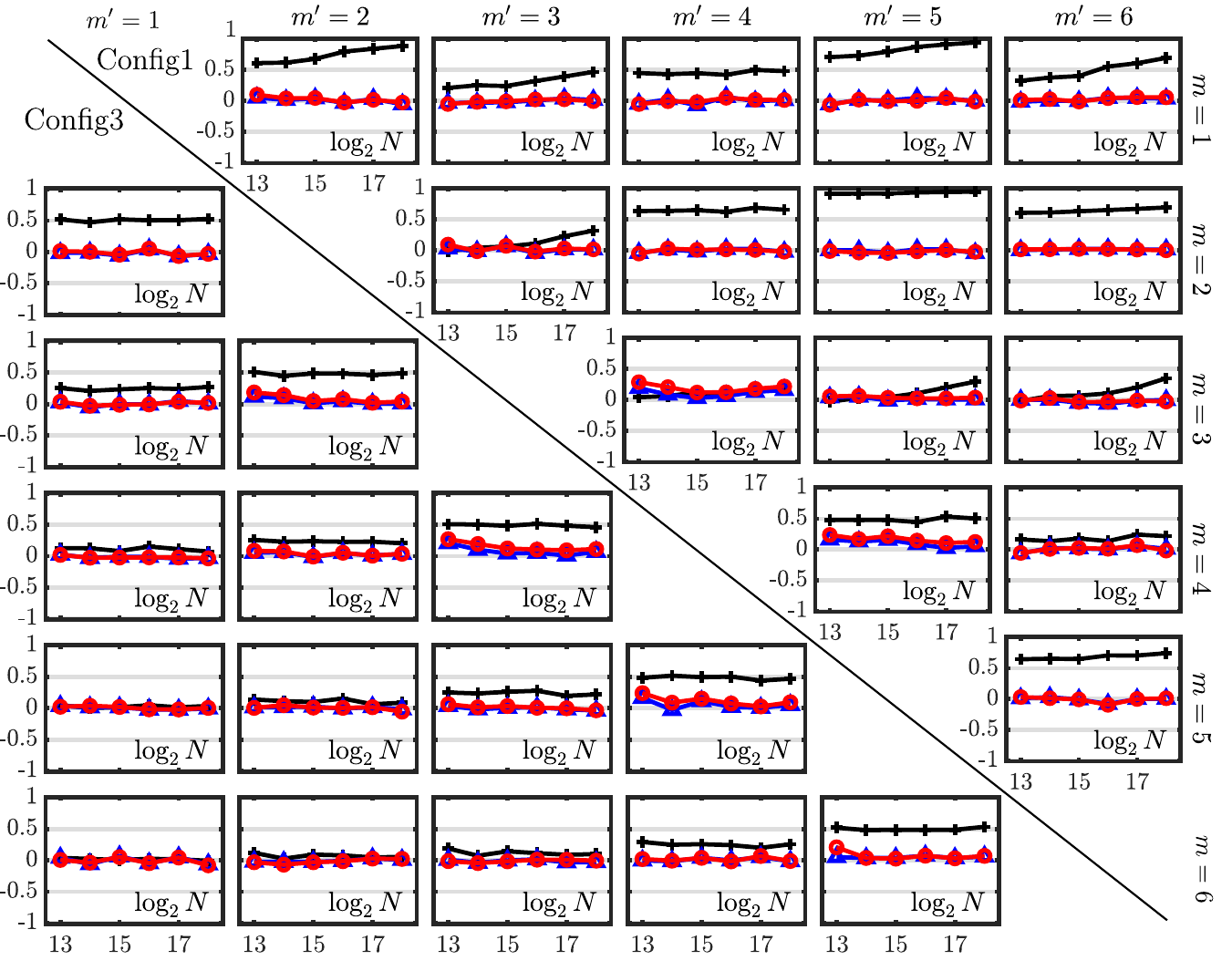}
\caption{\label{fig:CorrHm} {\bf Correlation structures of $\boldsymbol{\underline{\hat H}^{\rm M,bc}}$, $\boldsymbol{\underline{\hat H}^{\rm M}}$ and $\boldsymbol{\underline{\hat H}^{\rm U}}$.}
Correlation coefficients between entries $\hat H_m$ and $\hat H_{m'}$, $1 \leq m \neq m' \leq M$,
for $\underline{\hat H}^{\rm M,bc}$ (red), $\underline{\hat H}^{\rm M}$ (blue) and $\underline{\hat H}^{\rm U}$ (black), as functions of sample sizes $N$, for {\bf Config1} (upper right triangle) and {\bf Config3} (lower left triangle) .}
\end{figure}

We complement the estimation performance assessment by turning to the dependence structure amongst estimates -- the latter constitutes a key feature in the practical use of the estimators of the selfsimilarity parameters. Asymptotic normality, as well as finite sample size approximate normality, hint on the importance of the correlation structure amongst estimated selfsimilarity parameters. Fig.~\ref{fig:CorrHm} thus compares correlation coefficients between pairs $(\hat H_m,\hat H_{m'})$, $1 \leq m < m' \leq M$, as functions of the sample size $N$, for the three different estimators $\underline{\hat H}^{\rm M,bc}$, $\underline{\hat H}^{\rm M}$ and $\underline{\hat H}^{\rm U}$. The correlation structures are compared only for \textbf{Config1} and \textbf{Config3}, in upper and lower triangles respectively. The conclusions drawn about the other configurations are similar.

Fig.~\ref{fig:CorrHm} clearly indicates that most pairs $(\hat H^{\rm U}_m, \hat H^{\rm U}_{m'})$, $1 \leq m < m' \leq M$, remain significantly correlated, even for large sample sizes. By contrast, for the multivariate estimators, the pairs $(\hat H^{\rm M,bc}_m, \hat H^{\rm M,bc}_{m'})$ and $(\hat H^{\rm M}_m, \hat H^{\rm M}_{m'})$, $1 \leq m < m' \leq M$ are quasi-decorrelated, for all sample sizes. These observations partly explain why the overall contribution of the variance to the MSE was larger for the univariate estimators when compared to the multivariate estimators (cf.\ Fig.~\ref{fig:figperf} above).
These empirical observations are complemented and framed by the analytical calculations detailed in Sections~\ref{sec:appendixC}~and~\ref{sec:appendixD} of the supplementary material.
{These show the following striking fact. Suppose {\sc Condition (C0)} holds, and further assume that the weak correlation structure amongst wavelet coefficients of mfBm can be approximated into exact decorrelation. Then, the covariance matrices ${\mathbf \Sigma}_{\bf B}$ in Eq.~\eqref{e:scale_invariance_of_the_distribution} and ${\mathbf \Upsilon}_{\bf B}$ in Eq.~\eqref{e:H-hat_m_is_asymptotically_normal} of Theorem~\ref{t:scale_invariance_of_the_distribution} are well approximated, to the first order, by diagonal matrices.}
These empirical observations and analytical calculations reveal an important advantage of the mutivariate $\underline{\hat H}^{\rm M,bc}$ (and $\underline{\hat H}^{\rm M}$) against the univariate $\underline{\hat H}^{\rm U}$, even for nonmixing data: {namely, that the former consist of asymptotically Gaussian, weakly dependent random vectors. This is a major practical feature, e.g., in the design of tests for the equality of selfsimilarity parameters, which is of significant interest in many applications \cite{lucas2021bootstrap,lucas2022counting}.}

\begin{figure}[!t]
\centerline{
\includegraphics[width=\linewidth]{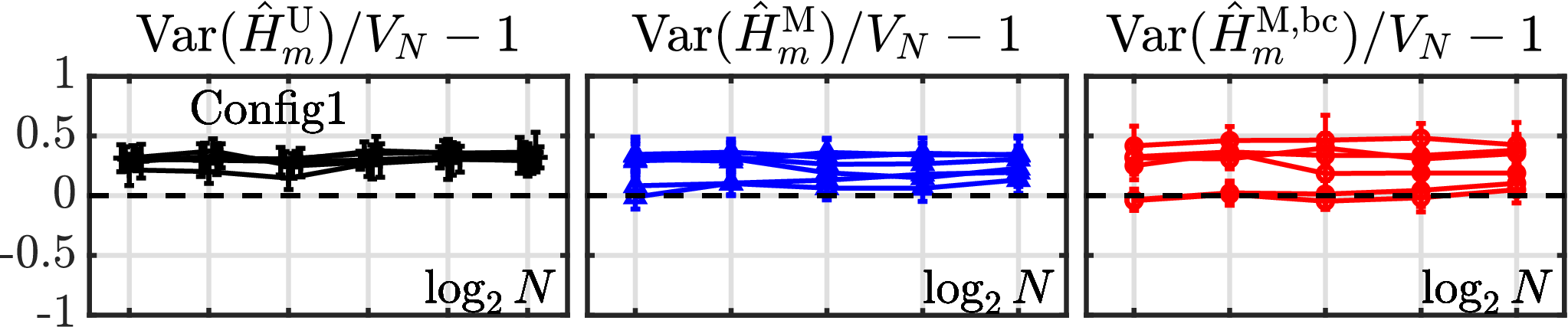}
}
\centerline{
\includegraphics[width=\linewidth]{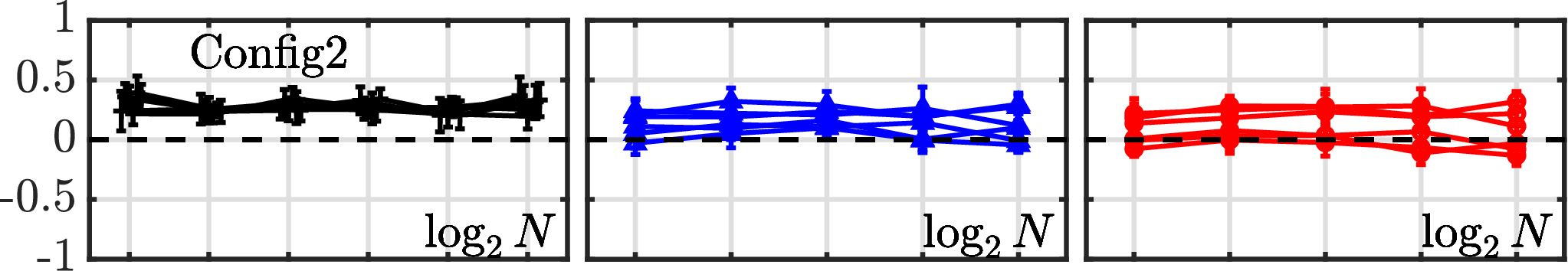}
}
\centerline{
\includegraphics[width=\linewidth]{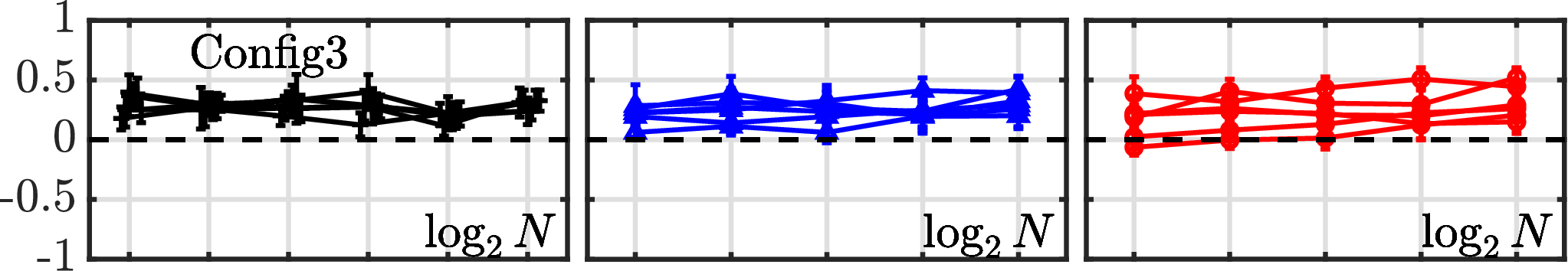}
}
\centerline{
\includegraphics[width=\linewidth]{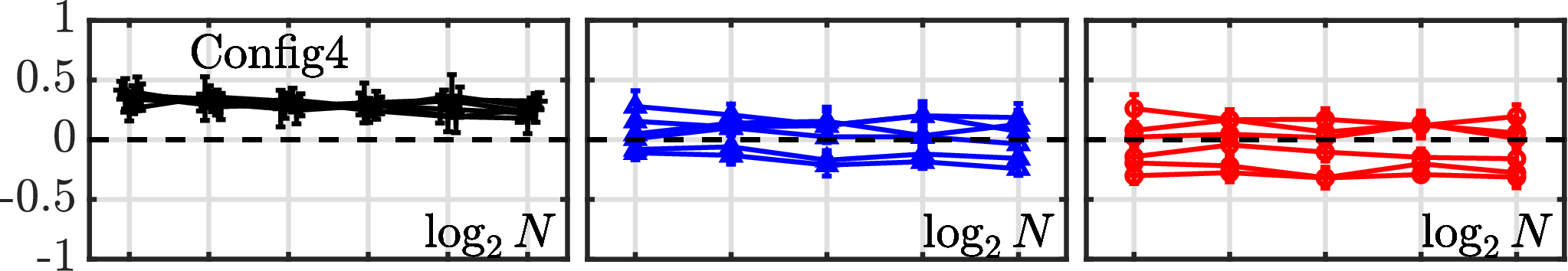}
}
\centerline{
\includegraphics[width=\linewidth]{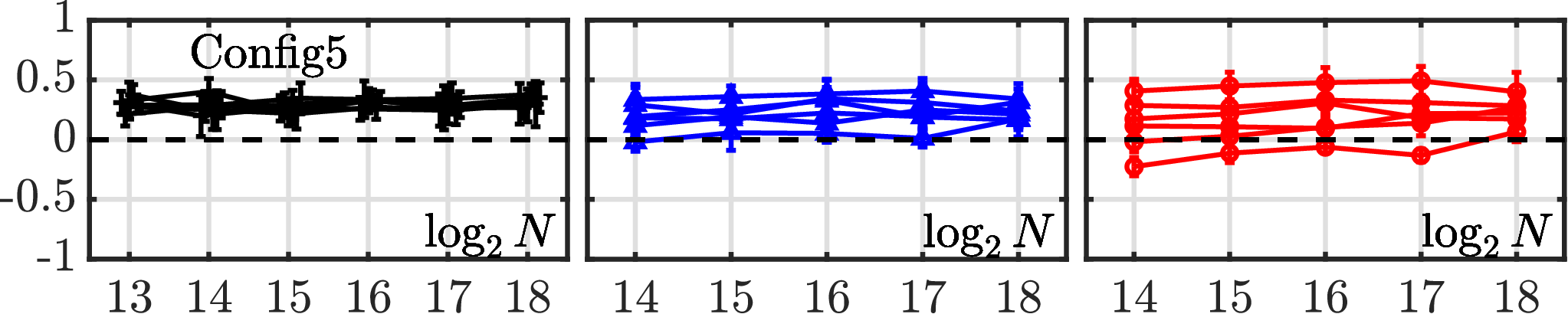}
}
\caption{\label{fig:varN} {\bf Approximation of the variance.} Relative differences between empirical variances of the (left, black) univariate estimator $\hat H^{\rm U}_m$, (middle, blue) multivariate estimator $\hat H^{\rm M}_m$ and (right, red) bias-corrected multivariate estimator $\hat H^{\rm M,bc}_m$ and the approximate variance $V_N$ (Eq.~\eqref{eq:varApprox}) with superimposed components $m=1,\ldots,6$ as functions of the sample size $N$.
}
\end{figure}

\subsubsection{Variance}

Besides approximate asymptotic decorrelation, the diagonal entries of the matrices
${\mathbf \Sigma}_{\bf B}$ and ${\mathbf \Upsilon}_{\bf B}$ remain to be discussed.
The calculations reported in Remark \ref{remB} in Section \ref{sec:appendixB}, and in Sections \ref{sec:appendixC}  and \ref{sec:appendixD} of the supplementary material show that these diagonal entries can be first-order approximated by a constant:
\begin{equation}
{\mathbf \Upsilon}_{\bf B} \simeq  2 \, {(\log_2 e )^2} \,  \mathbb{I} \makebox{ and }
{\mathbf \Sigma}_{\bf B} \simeq   \frac{(\log_2 e )^2}{2} \sum_{j=j_1^0}^{j_2^0} {w_j^2}{2^{j}} \, \, \mathbb{I}. \label{eq:varApproxB}
\end{equation}
Specifically, this implies that, under the approximation of decorrelation between wavelet coefficients of mfBm and under (C0), the finite-size variances of $\hat H^{\rm M,bc}_m$ and $\hat H^{\rm M}_m$, for every $m \in \{1,\ldots,M\}$, can be approximated as:
\begin{equation}
\label{eq:varApprox}
\mathrm{Var}(\hat H^{\rm M,bc}_m) \approx \mathrm{Var}(\hat H^{\rm M}_m) \approx \displaystyle \frac{(\log_2 e )^2}{2} \sum_{j=j_1}^{j_2} \frac{w_j^2}{n_{j}} := V_N.
\end{equation}
When $N \rightarrow +\infty$, one obviously has {$(N/a(N)) V_N \rightarrow $ $ \frac{(\log_2 e )^2}{2} \sum_{j=j_1^0}^{j_2^0} {w_j^2}{2^{j}} $.
Thus, in conclusion, $\mathrm{Var}(\hat H^{\rm M}_m)$ and $\mathrm{Var}(\hat H^{\rm M,bc}_m)$ decrease with sample size $N$ and depend on the range of regressions scales $j_1 $ and $j_2$. However, to the first order, these depend on none of the model parameters ${\mathbf \Sigma}$, ${\mathbf W}$ and  $\underline{H}$. This renders them comparable for all components.}
Interestingly, it was shown in ~\cite{wendt2017multivariate} that, for the univariate $\hat H^{\rm U}_m$ in nonmixing settings, the variance can be approximated to the first order by the same closed-form relation: $\mathrm{Var}(\hat H^{\rm U}_m) \approx V_N$.

To further assess the practical relevance of the finite sample size approximation in Eq.~\eqref{eq:varApprox}, Fig.~\ref{fig:varN} reports the relative difference
between the empirical variances of $\hat H^{\rm U}_m$, $\hat H^{\rm M}_m$ and $\hat H^{\rm M,bc}_m$ and $V_N$ as a function of the logarithm of the sample size $N$ for $m=1,\ldots,6$ and for the different configurations.
Fig.~\ref{fig:varN} (left column), first, shows that, in all configurations, the variance of the univariate estimator $\hat H^{\rm U}_m$ is approximated by $V_N$ with a $20\%$ relative error resulting from the non-exact decorrelation of eigenvalues between scales.
This first result extends the results reported in~\cite{wendt2017multivariate}, limited to nonmixing cases. Fig.~\ref{fig:varN}, second, confirms that the variances of the three estimators neither depend on $\underline{H}$ (\textbf{Config1} vs. \textbf{Config2}),
nor on ${\mathbf W}$ (\textbf{Config1} vs.\ \textbf{Config3}), nor on ${\mathbf \Sigma}$ (\textbf{Config1} vs. \textbf{Config5}).
Fig.~\ref{fig:varN}, third, confirms that the variances of the multivariate estimators are, on average, lower than those of the univariate estimator.
In summary, Fig.~\ref{fig:varN} shows that the closed-form approximation $V_N$ in Eq.~\eqref{eq:varApprox}
is valid in all configurations for the three estimators, even for small sample sizes.
This constitutes a valuable result for practical use in applications, e.g., when testing the equality of selfsimilarity parameters.

\section{Application to epileptic seizure prediction}
 \label{sec.appli}

Selfsimilarity has been broadly used to model temporal dynamics in macroscopic brain activity. The purpose of the modeling effort is either to study healthy brain mechanisms (cf. e.g., \cite{He2014,CIUCIU:2014:A,la2018self}), or to diagnose brain disorders and pathologies. Notably, there is a particular interest in epileptic seizure prediction \cite{domingues2019multifractal}.
Yet, most studies have so far remained based on univariate analysis of scale-free dynamics, while brain dynamics is naturally conducive to multivariate analysis.
The aim of this section is, thus, to comparatively assess the potential of multivariate and univariate selfsimilarity parameter estimations in the detection of 
ictal, preictal and interictal states against each other (i.e., time periods during seizure, prior to seizure onsets or distant from any seizure, respectively).

The multi-channel EEG recordings studied here were collected at the Boston Children's Hospital from pediatric patients with medically intractable seizures,
from the CHB-MIT Scalp EEG database \cite{goldberger2000physiobank}, available at \url{https://physionet.org/content/chbmit/1.0.0/}.
EEG recordings stem from the International 10-20 system of EEG electrode positions and nomenclature, last for at least one hour, and are sampled at 256Hz.
Multivariate analyses are performed for the 19 channels, within sliding 16 second-long wondows. 
Wavelet analysis is performed using the least asymmetric Daubechies2 mother wavelet.

Fig.~\ref{fig:structure function_patients} reports, for ictal, preictal and interictal windows, the logarithms of the diagonal entries of the wavelet spectrum ${\mathbf S}(2^j)$, $\log_2 {\mathbf S}_{m,m}(2^j)$, the log-eigenvalues $\log_2 \hat \lambda_m(2^j)$ and the proposed bias-corrected log-eigenvalues $\overline{\log_2 \lambda_m}(2^j)$, all of them as functions of scales $2^j$ for $m=1,\ldots,19$.
Fig.~\ref{fig:structure function_patients} shows linear behaviors for these functions across fine scales $1\leq 2^j \leq 2^7$, that hence indicate scale-free dynamics in these EEG time series, at frequencies above $1$Hz, in agreement with the literature on epilepsy,~cf. e.g., \cite{gadhoumi2015scale}.
These scale-free dynamics hold for the three different states and both component-wise (univariate, $\log_2 {\mathbf S}_{m,m}(2^j)$) and jointly (multivariate, $\log_2 \hat \lambda_m(2^j)$ and $\overline{\log_2 \lambda_m}(2^j)$).
The observed multivariate selfsimilarity constitutes per se a significant outcome of this work as selfsimilarity in EEG time series had so far only been observed component per component and not jointly.
Combined with literature suggesting that differences between preictal and interictal dynamics are related to frequencies ranging from $10$Hz to $85$Hz \cite{gadhoumi2015scale}, these observations suggest performing linear regressions for the estimation of $\underline{H}$ across scales ranging from $2^1 \leq 2^j \leq 2^4$.

\begin{figure}[!t]
\centerline{\includegraphics[width=\linewidth]{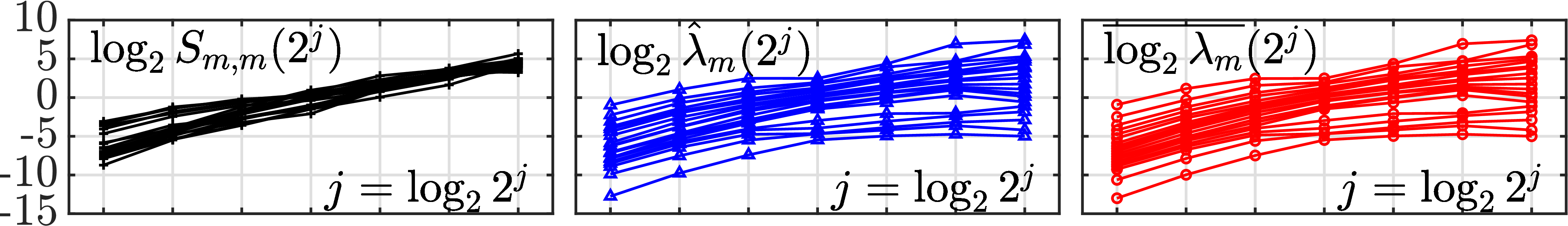}}
\centerline{\includegraphics[width=\linewidth]{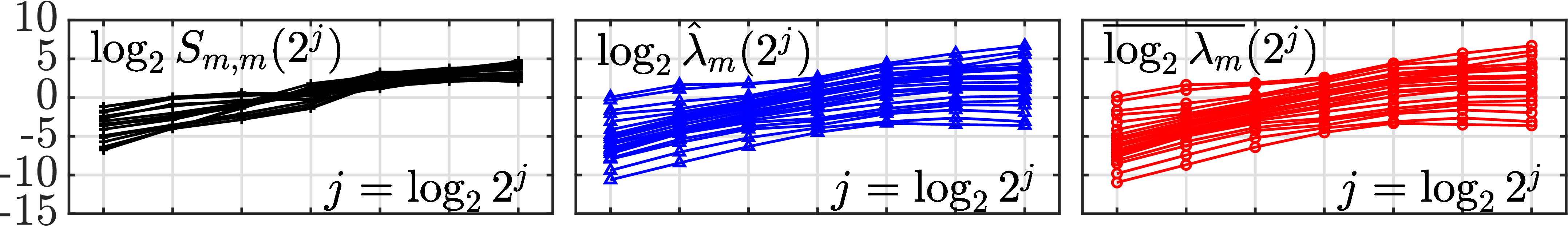}}
\centerline{\includegraphics[width=\linewidth]{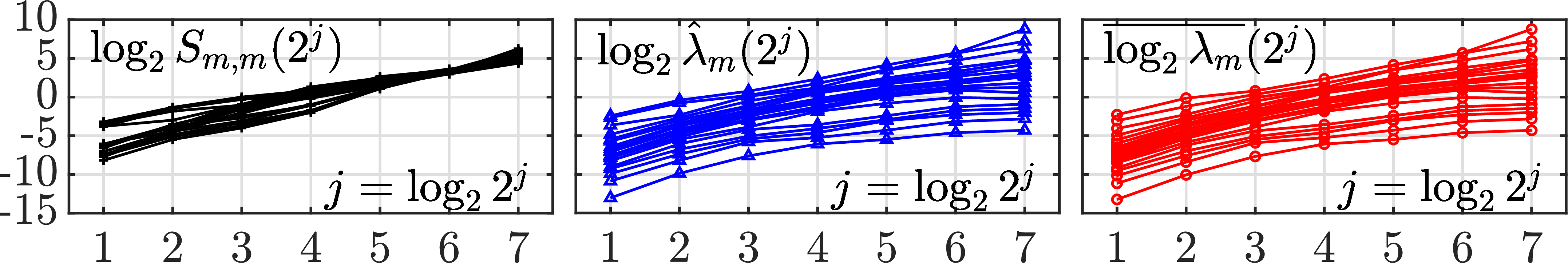}}
\caption{\label{fig:structure function_patients} {\bf Multivariate scale-free dynamics in EEG time series, for ictal, preictal and interictal  states.} Logarithms of the estimated diagonal entries $\log_2 {\mathbf S}_{m,m}(2^j)$ (left, black), estimated eigenvalues $\log_2 \hat \lambda_m(2^j)$ (middle, blue) and bias corrected eigenvalues $\overline{\log_2 \lambda_m}(2^j)$ (right, red) of the wavelet spectrum, for $m=1,\ldots,19$, for  ictal (top), preictal (middle),   interictal (bottom) windows.}
\end{figure}

\begin{figure}[!t]
\centerline{
\includegraphics[width=.335\linewidth]{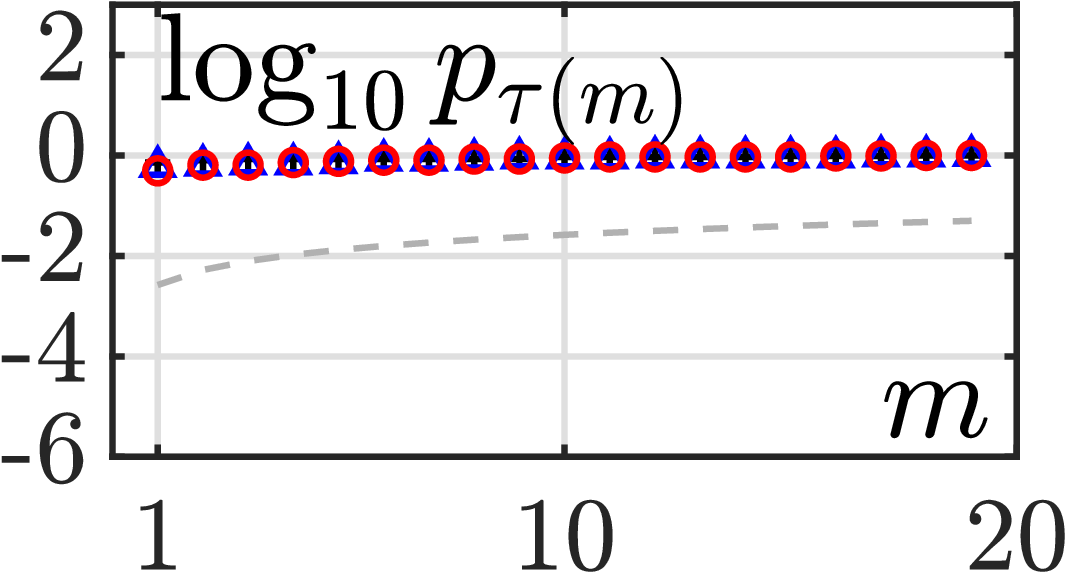}
\includegraphics[width=.3\linewidth]{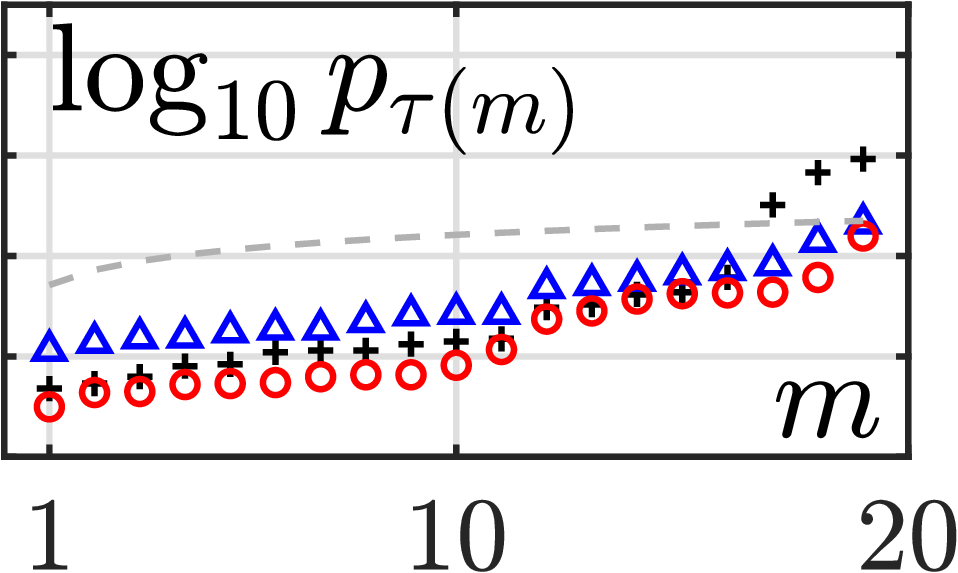}
\includegraphics[width=.3\linewidth]{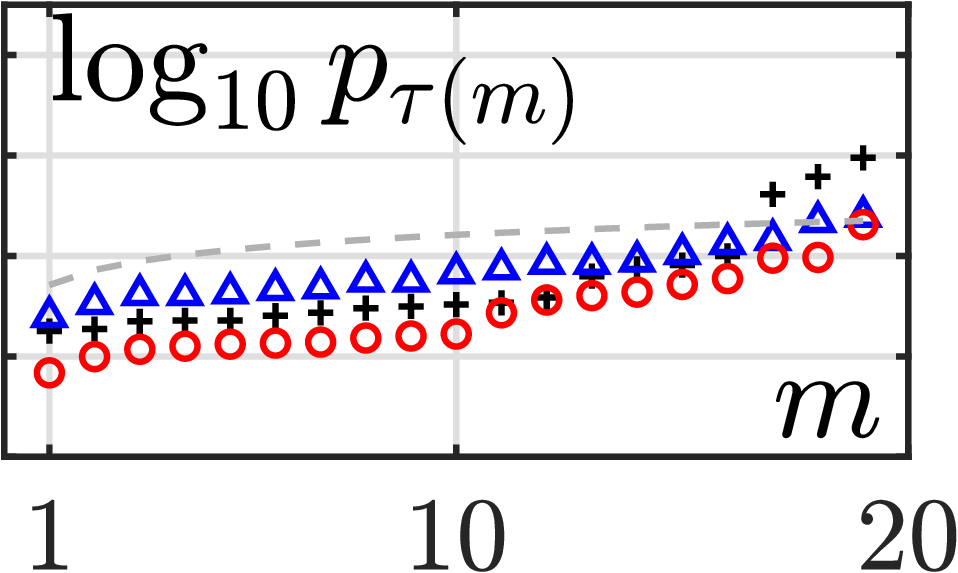}}
\caption{\label{fig:boxplot} {\bf Differences in multivariate selfsimilarities.}
Logarithms of the sorted p-values for Wilcoxon rank-sum tests between ictal / preictal distributions (left), ictal / interictal distributions (middle) and preictal / interictal distributions (right), for univariate estimates  $\hat H^{\rm U}_m$ (black), multivariate estimates $\hat H^{\rm M}_m$ (blue) and bias corrected multivariate estimates $\hat H^{\rm M,bc}_m$ (red), vs.\ Benjamini-Hochberg multiple hypothesis testing thresholds (grey dashed line).}
\end{figure}

Fig.~\ref{fig:boxplot} compares the (logarithms of) sorted Wilcoxon rank-sum test p-values $p_{\tau(m)}$ measuring the significance of the differences between the medians of ictal vs.\ preictal, ictal vs.\ interictal and preictal vs.\ interictal distributions, for  the univariate estimator $\underline{\hat H}^{\rm U} $ (Eq.~\eqref{eq:UEst}), the multivariate estimator $\underline{\hat H}^{\rm M} $ (Eq.~\eqref{eq:Mest}) and the bias-corrected multivariate estimator $\underline{\hat H}^{\rm M, bc} $ (Eq.~\eqref{eq:bcHurstEstimation}).
The p-values are ranked and compared to Benjamini-Hochberg multiple hypothesis thresholds~\cite{benjamini1995FDR}, computed at a false discovery rate $\alpha=0.05$.
For all three estimators, the medians of the distributions of estimated selfsimilarity parameters clearly differ
from ictal/preictal to interictal states but not from preictal to ictal states.
This indicates that the scale-free structure of the temporal dynamics is modified by forthcoming epileptic seizures but not within periods evolving to a seizure.
Moreover, the bias-corrected multivariate estimator $\underline{\hat H}^{\rm M,bc}$ shows stronger differences in selfsimilarity parameters between ictal or preictal states and interictal states, compared to either the univariate estimates $\underline{\hat H}^{\rm U}$ or the multivariate estimates $\underline{\hat H}^{\rm M}$.
This illustrates  the benefits of both taking into account cross-temporal dynamics in EEG time series analysis and correcting the repulsion bias in multivariate estimation for epilepsy seizure prediction.

These analyses are reported for a single patient only (Patient $22$).
Equivalent conclusions were drawn for all other patients.
Of importance, the distributions of selfsimilarity parameters are found to differ significantly from one patient to another. This demands a single-patient procedure to detect ictal / preictal from interictal states.
Interested readers are referred to the in-depth study carried out in~\cite{lucas2023epileptic}.

\section{Conclusions and Perspectives}
 \label{sec.conclusion}

In this work, we have, first, proposed a multivariate selfsimilarity model, defined as a linear mixture of correlated fractional Brownian motions, well-suited to adjust the versatility of real-world multivariate data, and formally related it to a former mathematical model.
Second, we have turned the formal proposition of wavelet spectrum eigenvalue-based estimation for the vector of selfsimilarity parameters proposed in~\cite{abry2018waveleta,abry2018waveletb}, into a versatile and reliable bias-corrected estimation procedure that can actually be used on real-world data.
Its performance was studied mathematically in terms of asymptotics and by numerical simulations for finite-size samples.
This permitted showing that the proposed multivariate estimators work in a broader context (mixing) than the univariate one and still performs better (decorrelation) in the nonmixing case, where the univariate estimator can be used.
Third, its potential was illustrated at work on real multi-channel EEG data in a context of epileptic seizure prediction.
A documented toolbox implementing estimation is publicly available at \url{github.com/charlesglucas/ofbm_tools}.

Future investigations include incorporating a multivariate time-scale bootstrap procedure so as to assess confidence levels in estimation from a single realization of data.
Also, large-dimensional asymptotics will be studied both theoretically and numerically to account for situations where the number of components grows as fast as sample size, of major practical importance in real-world applications.

\bibliographystyle{IEEEtran}
\bibliography{ofbmTSP2023}

\appendix

Fix $\ell \in \bbN$. Throughout the appendices, ${\mathcal S}(\ell,\bbR)$ and ${\mathcal S}_{\geq 0}(\ell,\bbR)$ denote, respectively, the space of symmetric and the cone of symmetric positive semidefinite $\ell \times \ell$ matrices. The notation $\textnormal{vec}_{{\mathcal S}}$ denotes the projection on the free entries of a symmetric matrix. Also, $\pi_{i_1,i_2}$ and $\pi_{i_1,i_2,\hdots,i_\ell}$, $\ell \in \bbN$, denote the projection of a matrix or vector, respectively, on its entries $i_1,i_2$ and $i_1,i_2,\hdots,i_\ell$. We express any vector ${\mathbf v}$ entry-wise as $\{{\mathbf v}_{\ell}\}_{\ell}$.
Recall that $n_{a,j} = \frac{N}{a(N)2^j}$, $ j = j^{0}_1,\hdots,j^{0}_2$.\\

\subsection{Proof of consistency}
\label{sec:appendixA}

\begin{proof}
We only show $(i)$. Fix $j \in \{j_1^0,\ldots,j_2^0\}$, $\tau \in \{1,\ldots,2^{j^0_2-j}\} $ and $m \in \{1,\ldots,M\}$. Suppose, for the moment, that the limit \eqref{e:lim_n_a*lambda/a^(2h+1)} holds. Recall that $\plim_{N \rightarrow \infty}$ denotes limit in probability. Writing $\widetilde a(N) = a(N)2^j$,
\begin{equation}
\begin{aligned}
\xi_m(2^j) & =\plim_{N \rightarrow \infty}\frac{\hat \lambda_{m}^{(\tau)}(a(N)2^j)}{a(N)^{2H_{m}+ 1}} \\
&= \plim_{N \rightarrow \infty}\frac{\hat \lambda_{m}^{(\tau)}\big(\widetilde a(N)\big)}{(\widetilde a(N)2^{-j})^{2H_{m}+ 1}} \\
&= 2^{j(2H_{m}+1)}\xi_m(1).
\end{aligned}
\end{equation}
This establishes the scaling relationship in \eqref{e:lim_n_a*lambda/a^(2h+1)}.

We now need to show the convergence in \eqref{e:lim_n_a*lambda/a^(2h+1)}. The argument follows from a simple adaptation of the technique for establishing Theorem 3.1 in \cite{abry:boniece:didier:wendt:2022:prob}, which originally pertains to a high-dimensional context.

In fact, in that theorem the measurements are given by a high-dimensional signal-plus-noise stochastic process of the form
$Y(t) = {\mathbf P}(N) X(t) + Z(t)$, where, for $p \geq M$, $Y(t), Z(t) \in \bbR^p$, ${\mathbf P}(N) \in {\mathcal M}(p,M,\bbR)$ and $X(t) \in \bbR^M$. $X(t)$ is the latent fractional stochastic process (the ``signal"), $Z(t)$ is the ``noise" component, and ${\mathbf P}(N)$ is a full rank matrix of coordinates. So, first assume the dimension $p$ is fixed and set to $M$. Now suppose $X(t)$ is an ofBm under the conditions of Theorem \ref{t:H-hat_m_is_consistent}. Next set the matrix ${\mathbf P}(N)$ to the $M \times M$ identity matrix. In addition, suppose the noise term $Z(t)$ is identically zero. In order to establish \eqref{e:lim_n_a*lambda/a^(2h+1)}, it now suffices to follow the argument of the proof of Theorem 3.1 in \cite{abry:boniece:didier:wendt:2022:prob}. The convergence \eqref{e:lim_n_a*lambda(EW)/a^(2h+1)} can be established by an analogous argument. This shows $(i)$.
\end{proof}

\begin{remark}
\label{remA}
The framework for proving Theorem \ref{t:H-hat_m_is_consistent} provides a great deal of information about the asymptotic rescaled eigenvalues $\xi_m(2^j)$, $m \in \{1,\hdots,M\}$, as in \eqref{e:lim_n_a*lambda/a^(2h+1)}. More specifically, for some function $\varphi$ and some vector $\widetilde{{\mathbf w}}$, we can write
\begin{equation}\label{e:plim_lambda_m/a^(2Hm+1)=varphi(tau)}
\plim_{N \rightarrow \infty}\frac{\hat{\lambda}^{(\tau)}_{m}(a(N)2^j)}{a(N)^{2H_m+1}} = \xi_m(2^j) = \varphi(\widetilde{{\mathbf w}}).
\end{equation}
In order to describe $\varphi(\widetilde{{\mathbf w}})$, we need a few definitions and constructs. So, for the fixed $m$, define the sets of indices $\mathcal{I}_- := \{\ell: H_\ell < H_m\}$ and
\begin{equation} \label{e:def_indexsets}
\mathcal{I}_0 := \{\ell: H_\ell = H_m\} \textnormal{ \textit{and} } \mathcal{I}_+ := \{\ell: H_\ell > H_m\}.
\end{equation}
Note that $\mathcal{I}_-$ and $\mathcal{I}_+$ are possibly empty. Also, let
\begin{equation}\label{e:r1,r2,r3}
r_1:= \textnormal{card}(\mathcal{I}_-), \hspace{1mm}r_2:= \textnormal{card}(\mathcal{I}_0) \geq 1, \hspace{1mm} r_3:= \textnormal{card}(\mathcal{I}_+)
\end{equation}
be their respective cardinalities. For every $\tau = 1,\hdots, 2^{j^0_2-j}$, define the auxiliary wavelet random matrix
\begin{equation}\mathbf{\widehat{B}}^{(\tau)}_a(2^j):= {\mathbf W}^{-1}a^{-\HH-\frac{1}{2}\mathbb{I}}{\mathbf S}^{(\tau)}\big(a2^j\big)a^{-\HH^*-\frac{1}{2}\mathbb{I}}({\mathbf W}^*)^{-1},
\end{equation}
where $a \equiv a(N)$. Also, let $\mathbf{B}(2^j) = \bbE \mathbf{B}_a(2^j)$, which neither depends on $\tau$ nor on $a$. Recast
\begin{equation}\label{e:B(2^j)=(B_i,ell)}
\widehat {\mathbf B}^{(\tau)}_a(2^j) =(\widehat {\mathbf B}_{i\ell})_{i,\ell=1,2,3}, \quad {\mathbf B}(2^j) =(\mathbf B_{i\ell})_{i,\ell=1,2,3}.
\end{equation}
In \eqref{e:B(2^j)=(B_i,ell)}, ${\mathbf B}_{i\ell}$ and $\widehat {\mathbf B}_{i\ell}$ denote blocks of size $r_i\times r_\ell$, and the scale $2^j$ is omitted wherever doing so is unambiguous. Similarly, define
\begin{equation}\label{e:P=(P1(N)_P2(N)_P3(N))}
\mathbf W  = \big(\mathbf W_1 ~\mathbf W_2 ~\mathbf W_3\big),
\end{equation}
where, for $i = 1,2,3$, $\mathbf W_i \in {\mathcal M}(p,r_i,\bbR)$, and each column vector is denoted by $\mathbf w_\ell$, $\ell= 1,\hdots,r$.
Let ${\boldsymbol \Pi}_3 \in {\mathcal M}(r,\bbR)$ be the projection matrix onto $\textnormal{Ker}\{{\mathbf W}^{*}_{3}\}$. Also, let
${\mathbf M} = {\mathbf B}_{22} - {\mathbf B}_{23}{\mathbf B}^{-1}_{33}{\mathbf B}_{32} \in {\mathcal S}(r_2,\bbR)$ and ${\boldsymbol \Lambda} =
{\boldsymbol \Pi}^*_3 {\mathbf W}_2 {\mathbf M} {\mathbf W}^*_2{\boldsymbol \Pi}_3\in {\mathcal S}_{\geq 0}(r,\bbR)$. For unit vectors ${\mathbf u} \in \bbR^r$, define the function
\begin{equation}\label{e:varphi(u)=varphi(x,u)}
\varphi({\mathbf u}) := {\mathbf u}^* {\boldsymbol \Lambda}{\mathbf u}.
\end{equation}
Then, for $\varphi$ as in \eqref{e:varphi(u)=varphi(x,u)}, the proof of Proposition 5.1 and Lemma B.1 in \cite{abry:boniece:didier:wendt:2022:prob} imply that there exists a possibly random vector
\begin{equation}\label{e:w_in_Lambda}
\widetilde{{\mathbf w}} = \widetilde{{\mathbf w}}(\omega) \in \textnormal{span}\{{\boldsymbol \Lambda}\} = \textnormal{span}\{{\mathbf W}_2, {\mathbf W}_3\} \cap \textnormal{Ker}\{{\mathbf W}^{*}_3\}
\end{equation}
based on which we can express the limit \eqref{e:plim_lambda_m/a^(2Hm+1)=varphi(tau)} as $\varphi(\widetilde{{\mathbf w}})$ (\textbf{n.b.}: $\xi_m(2^j)$ is \textnormal{deterministic}). In particular, relations \eqref{e:varphi(u)=varphi(x,u)} and \eqref{e:w_in_Lambda} imply that $\varphi({\mathbf w})$ -- and, hence, $\xi_m(2^j)$ -- only depends on the matrices ${\mathbf B}_{22}$, ${\mathbf B}_{23}$, ${\mathbf B}_{33}$, ${\mathbf W}_2$ and ${\mathbf W}_3$. Consequently, the asymptotic rescaled function in \eqref{e:lim_n_a*lambda/a^(2h+1)} only depends on parameters associated with the index sets ${\mathcal I}_0$ and ${\mathcal I}_+$ in \eqref{e:def_indexsets}.
\end{remark}

\subsection{Proof of asymptotic normality}
\label{sec:appendixB}

In the proof of Theorem \ref{t:scale_invariance_of_the_distribution}, we write $a \equiv a(N)$ whenever notationally convenient. Also, we make use of the following notation. Let $\{X_N\}_{N \in \bbN}$, $\{Y_N\}_{N \in \bbN}$ be two sequences of random variables. We write $X_N \stackrel{d}\sim Y_N$, $N \rightarrow \infty$, when the two sequences have the same limit in distribution, which is assumed to be well defined.

\begin{proof}
First, we show $(i)$. Fix $j \in \{j^0_1,\hdots,j^0_2\}$. The proof of \eqref{e:scale_invariance_of_the_distribution} is by means of an adaptation of the proof of Theorem 3.2 in \cite{abry:boniece:didier:wendt:2022:prob}. In fact, let $\widehat {\mathbf B}_a(2^j)$ and ${\mathbf B}(2^j)$ be as in \eqref{e:B(2^j)=(B_i,ell)}. Under conditions (OFBM1-3), as a consequence of Theorem 3.1 in \cite{abry2018waveleta},
\begin{equation}\label{e:sqrt(n_a,j2)(B^w-hat-B(2^j))->asympt_Gauss}
\Big\{\sqrt{n_{a,j^0_2}}\hspace{0.5mm}\textnormal{vec}_{{\mathcal S}}\big(\mathbf{\widehat{B}}^{(\tau)}_a(2^j)-\mathbf{B}(2^j)\big)\Big\}_{\tau = 1,\hdots, 2^{j^0_2-j}} \stackrel{d}\rightarrow {\mathcal N}(0,\mathbf{\Sigma}_{\mathbf{B}})
\end{equation}
as $n \rightarrow \infty$ for some $\mathbf{\Sigma}_{\mathbf{B}} \in {\mathcal S}_{\geq 0}\big(2^{j^0_2-j}\cdot M(M+1)/2,\bbR\big)$. In particular,
\begin{equation}\label{e:B^w-hat->B(2^j)}
\mathbf{\widehat{B}}^{(\tau)}_a(2^j) - \mathbf{B}(2^j) \stackrel{\bbP} \rightarrow 0, \quad \tau = 1,\hdots, 2^{j^0_2-j}.
\end{equation}
So, fix $m \in \{1,\hdots, M\}$. Recall that, for any matrix $\mathds{M} \in {\mathcal S}(M,\bbR)$, the differential of a \textit{simple} eigenvalue $\lambda_m(\mathds{M})$ exists in a connected vicinity of $\mathds{M}$ and is given by
\begin{equation}\label{e:d_lambda}
d \lambda_m(\mathds{M}) = \mathbf{u}^*_m \{d\hspace{0.25mm}\mathds{M}\}\mathbf{u}_m,
\end{equation}
where $\mathbf{u}_m$ is a unit eigenvector of $\mathds{M}$ associated with $\lambda_m(\mathds{M})$ (see \cite{magnus:1985}, p.\ 182, Theorem 1). Now, for every $\mathbf{B} \in {\mathcal S}_{\geq 0}(M,\bbR)$, define the function $f_{N,m}$ by
\begin{equation}\label{e:f_N,m}
f_{N,m}(\mathbf{B}) = \log_2 \lambda_m \Big( \frac{{\mathbf W} a^{\diag(\Hv)} \mathbf{B}
a^{\diag(\Hv)} {\mathbf W}^*}{a^{2H_m}}\Big).
\end{equation}
For any fixed $\delta > 0$, let ${\mathcal O}_{\delta} = \{\mathbf{B} \in {\mathcal S}_{\geq 0}(M,\bbR): \|\mathbf{B}-\mathbf{B}(2^j)\| < \delta \}$. By Proposition 3.1 in \cite{abry2018waveleta}, $\det \mathbf{B}(2^j) > 0$. Thus, an application of the proof of Theorem 3.2 in \cite{abry:boniece:didier:wendt:2022:prob} implies that there exists some $\delta_0 > 0$ such that, for large enough $N$ and for any $\mathbf{B} \in {\mathcal O}_{\delta_0}$, the mean value theorem-type relation
\begin{equation}\label{e:by_the_MVT_AbryDidier2018_wavereg}
\begin{split}
&f_{N,m}(\mathbf{B})-f_{N,m}(\mathbf{B}(2^j))  \\
& = \sum_{1 \leq i_1 \leq i_2 \leq M}\frac{\partial}{\partial b_{i_1,i_2}}
f_{N,m}(\mathbf{\breve{B}}) \cdot \pi_{i_1,i_2}\big(\mathbf{B} - \mathbf{B}(2^j)\big)
\end{split}
\end{equation}
holds for some matrix $\mathbf{\breve{B}} \in {\mathcal S}_{> 0}(M,\bbR)$ lying in a segment connecting $\mathbf{B}$ and $\mathbf{B}(2^j)$ across ${\mathcal S}_{> 0}(M,\bbR)$. Define the event $A_N = \{\omega:\mathbf{\widehat{B}}^{(\tau)}_a(2^j) \in {\mathcal O}_{\delta_0}, \hspace{0.5mm}\tau = 1,\hdots,2^{j^0_2-j} \}$. By \eqref{e:B^w-hat->B(2^j)}, $\bbP(A_N) \rightarrow 1$, $N \rightarrow \infty$. Hence, by \eqref{e:by_the_MVT_AbryDidier2018_wavereg}, with probability going to 1, for large enough $N$, for $\tau = 1,\hdots,2^{j^0_2-j}$ and $m = 1,\hdots,M$, expansion \eqref{e:by_the_MVT_AbryDidier2018_wavereg} holds with $\mathbf{B}=\mathbf{\widehat{B}}^{(\tau)}_a(2^j), \mathbf{\breve{B}}= \mathbf{\breve{B}}^{(\tau)}_a(2^j) \in {\mathcal S}_{> 0}(M,\bbR)$. In addition, we can express, for any $i_1,i_2 \in \{1,\hdots,M\}$,
\begin{equation}\label{e:d/d_bi1,i2_fN,m}
\begin{split}
& \frac{\partial}{\partial b_{i_1,i_2}}
f_{N,m}(\mathbf{\breve{B}}^{(\tau)}_a(2^j))  \\
&\qquad = \frac{1}{\mathrm{ln} \, 2}
\frac{a^{2H_m+1}}{\lambda_{m}\big(\mathbf{\breve{S}}^{(\tau)}(a2^j)\big)}\frac{\partial}{\partial b_{i_1,i_2}}
\lambda_{m}\Big(\frac{\mathbf{\breve{S}}^{(\tau)}(a2^j)}{a^{2H_m+1}}\Big),
\end{split}
\end{equation}
where $\mathbf{\breve{S}}^{(\tau)}(a2^j) := {\mathbf W} a^{\diag(\Hv)}\mathbf{\breve{B}}^{(\tau)}_a(2^j)a^{\diag(\Hv)}{\mathbf W}^*$. Expressions \eqref{e:by_the_MVT_AbryDidier2018_wavereg} and \eqref{e:d/d_bi1,i2_fN,m} imply that, as $N \rightarrow \infty$,
\begin{equation}\label{e:log2_sqrt(n_a,j)(fN,m(B-hat_a(2^j))-fN,m(B_a(2^j)))}
\begin{aligned}
&\mathrm{ln} \, 2 \hspace{0.5mm} \sqrt{n_{a,j^0_2}}\big(f_{N,m}(\widehat{{\mathbf B}}^{(\tau)}_{a} (2^j)) - f_{N,m}({\mathbf B}_a(2^j)) \big) \\
& \quad \stackrel{d}\sim \frac{\sqrt{n_{a,j^0_2}}}{\xi_{m}(2^j)}\sum^{M}_{i_1,i_2 = r_1+1}\langle {\mathbf w}_{i_1}(N),{\mathbf u}_m(N)\rangle a^{H_{i_1}-H_m} \\
& \cdot  \langle {\mathbf w}_{i_2}(N),{\mathbf u}_m(N)\rangle a^{H_{i_2}-H_m}\hspace{1mm} \pi_{i_1,i_2}\big(\widehat{{\mathbf B}}^{(\tau)}_{a}(2^j) - {\mathbf B}_a(2^j)\big) \\
& \quad \stackrel{d}\sim \frac{\sqrt{n_{a,j^0_2}}}{\xi_{m}(2^j)}\Big\{\sum^{r_1+r_2}_{i_1,i_2 = r_1+1} + 2 \sum^{r_1+r_2}_{i_1 = r_1+1}\sum^{M}_{i_2 = r_1+r_2+1} \\
&\qquad+ \sum^{M}_{i_1,i_2 = r_1+r_2+1}\Big\} \langle {\mathbf w}_{i_1}(N),{\mathbf u}_m(N)\rangle a^{H_{i_1}-H_m} \\
\cdot \langle & {\mathbf w}_{i_2}(N),{\mathbf u}_m(N)\rangle a^{H_{i_2}-H_m}\hspace{0.5mm}\pi_{i_1,i_2}\big(\widehat{{\mathbf B}}^{(\tau)}_{a}(2^j) - {\mathbf B}(2^j)\big). 
\end{aligned}
\end{equation}
However, by Proposition B.1 in \cite{abry:boniece:didier:wendt:2022:prob}, under condition \eqref{e:H_ell1_neq_H_ell2_=>_xi_ell1_neq_xi_ell2} we can suppose that there exists a deterministic unit vector ${\mathbf u}_m$ such that
\begin{equation}\label{e:u_m(N)->u_m}
{\mathbf u}_m(N) \stackrel{\bbP}\rightarrow {\mathbf u}_m.
\end{equation}
So, let
\begin{equation}\label{e:gamma-m_vec}
{\boldsymbol \gamma}_m = \{ {\boldsymbol \gamma}_{m,i}\}_{i = 1,\hdots,r} :=\plim_{N \rightarrow \infty}\mathbf W^*{\mathbf u}_{m}(N) \in \bbR^{r}.
\end{equation}
Also, by Lemma B.8 in \cite{abry:boniece:didier:wendt:2022:prob}, we can assume the existence of the limiting vector
\begin{equation}
\begin{aligned}\label{e:x_*,m}
{\mathbf x}_{*,m} = & \{{\mathbf x}_{*,m,i}\}_{i \in {\mathcal I}_+}:= \\
& \plim_{N \rightarrow \infty}\big\{\langle \mathbf w_i ,{\mathbf u}_{m}(N)\rangle \cdot a(N)^{H_{i}-H_{m}} \big\}_{i \in {\mathcal I}_+} \in \bbR^{r_3}.
\end{aligned}
\end{equation}
Therefore, \eqref{e:log2_sqrt(n_a,j)(fN,m(B-hat_a(2^j))-fN,m(B_a(2^j)))} has the same limit in distribution as
\begin{equation}\label{e:rescaled_logeig_gamma_x}
\begin{aligned}
 \frac{\sqrt{n_{a,j^0_2}}}{\xi_{m}(2^j)} & \Bigg\{\sum^{r_1+r_2}_{i_1,i_2 = r_1+1} {\boldsymbol \gamma}_{m,i_1} \cdot {\boldsymbol \gamma}_{m,i_2} \\
+ 2 \sum^{r_1+r_2}_{i_1 = r_1+1}& \sum^{M}_{i_2 = r_1+r_2+1} {\boldsymbol \gamma}_{m,i_1}  \cdot {\mathbf x}_{*,m,i_2}\\
+\sum^{M}_{i_1,i_2 = r_1+r_2+1} {\mathbf x}_{*,m,i_1}\cdot & {\mathbf x}_{*,m,i_2}  \hspace{0.5mm}\Bigg\} \cdot \pi_{i_1,i_2}\big(\widehat{{\mathbf B}}^{(\tau)}_{a}(2^j) - {\mathbf B}(2^j)\big).
\end{aligned}
\end{equation}
Now note that \eqref{e:rescaled_logeig_gamma_x} further holds for all $\tau$, $m$ and $j$, and also that these expansions start from the same sequence of asymptotically Gaussian random matrices $\widehat{{\mathbf B}}^{(\tau)}_a(2^j)$ (see \eqref{e:sqrt(n_a,j2)(B^w-hat-B(2^j))->asympt_Gauss}). Hence, the jointly asymptotically normal limit \eqref{e:scale_invariance_of_the_distribution} holds. Thus, $(i)$ is proven.

Statement \eqref{e:H-hat_m_is_asymptotically_normal} can be proven by naturally adapting the proof of Theorem 3.2 in \cite{abry:boniece:didier:wendt:2022}, originally developed for the high-dimensional context (see the related comments in the proof of Theorem \ref{t:H-hat_m_is_consistent} in this paper). In particular, the definitions \eqref{e:a(N)} and \eqref{e:varpi_parameter} are applied in a mathematically sharp way. This establishes $(ii)$.
\end{proof}

\begin{remark}
\label{remB}
For a single $j$, we can show that the limit in distribution \eqref{e:scale_invariance_of_the_distribution} of wavelet log-eigenvalues is \textit{scale-free}. 
Consider the vector ${\boldsymbol \gamma}_m$ as in \eqref{e:gamma-m_vec} as well as the vector
\begin{equation}\label{e:varrho-m}
{\boldsymbol \varrho}_m := \mathbf B^{-1}_{33}(1)\mathbf B_{32}(1)\pi_{r_1+1,\hdots,r_1+r_2}({\boldsymbol \gamma}_m) \in \bbR^{r_3}.
\end{equation}
Assume for the moment that
\begin{equation}\label{e:Taylor_last_equivalence}
\begin{aligned}
\sqrt{n_{a,j^0_2}}\big\{ \log_2\hat \lambda_{m}^{(\tau)}(a(N)2^j)& - \log_2 \lambda_{m}^{(\tau)}(a(N)2^j) \big\} \\
\stackrel{d}\sim \frac{1}{\ln 2}\frac{\sqrt{n_{a,j^0_2}}}{\xi_{m}(1)} &  \Bigg\{\sum^{r_1+r_2}_{i_1,i_2 = r_1+1} {\boldsymbol \gamma}_{m,i_1}\cdot {\boldsymbol \gamma}_{m,i_2}  \\
+ 2 \sum^{r_1+r_2}_{i_1 = r_1+1} & \sum^{M}_{i_2 = r_1+r_2+1}  {\boldsymbol \gamma}_{m,i_1}\cdot  {\boldsymbol \varrho}_{m,i_2} \\
+\sum^{M}_{i_1,i_2 = r_1+r_2+1} {\boldsymbol \varrho}_{m,i_1}& \cdot {\boldsymbol \varrho}_{m,i_2}\Bigg\} \cdot \pi_{i_1,i_2}\big(\widehat{{\mathbf B}}^{(\tau)}_{a}(1) - {\mathbf B}(1)\big).
\end{aligned}
\end{equation}
In light of \eqref{e:sqrt(n_a,j2)(B^w-hat-B(2^j))->asympt_Gauss}, it is clear that \eqref{e:Taylor_last_equivalence} converges to a Gaussian law that does not depend on $j$, as claimed.

To establish \eqref{e:Taylor_last_equivalence}, recast entry-wise ${\mathbf B}(2^j)=: \big(b_{\ell_1,\ell_2}(2^j)\big)_{\ell_1,\ell_2 = 1,\hdots,r}$ (cf.\ \eqref{e:B(2^j)=(B_i,ell)}). By Proposition 3.1, (P5), in \cite{abry2018waveleta},
\begin{equation}\label{e:B(2^j)=scaling}
{\mathbf B}(2^j)= \big((2^j)^{H_{\ell_1}+H_{\ell_2}+1}\hspace{0.5mm}b_{\ell_1,\ell_2}(1)\big)_{\ell_1,\ell_2 = 1,\hdots,r}.
\end{equation}
Bearing in mind \eqref{e:u_m(N)->u_m} and \eqref{e:varrho-m}, by the proofs of Lemma B.8 and Proposition 5.1 in \cite{abry:boniece:didier:wendt:2022:prob}, and by \eqref{e:B(2^j)=scaling}, we can recast \eqref{e:x_*,m} in the form
\begin{equation}\label{e:x*,m_reexpressed}
\begin{aligned}
{\mathbf x}_{m,*} & = - {\mathbf B}^{-1}_{33}(2^j){\mathbf B}_{32}(2^j) {\mathbf W}^*_2{\mathbf u}_m \\
& = - (2^j)^{H_{m}}
\textnormal{diag}\big((2^j)^{- H_{r_1+r_2+1}},\hdots,(2^j)^{-H_{r}} \big)\\
& \qquad \mathbf B^{-1}_{33}(1)\mathbf B_{32}(1)\mathbf W_2^*{\mathbf u}_m\\
 & = - (2^j)^{H_{m}} (2^j)^{-\textnormal{diag}(H_{r_1+r_2+1},\hdots,H_{r})} {\boldsymbol \varrho}_m.
\end{aligned}
\end{equation}
Rewrite, entry-wise,
\begin{equation}\label{e:Bhat^(tau)}
\widehat{{\mathbf B}}^{(\tau)}_a(2^j) = \big( \widehat{b}^{(\tau)}_{\ell_1,\ell_2}(2^j)\big)_{\ell_1,\ell_2=1,\hdots,r}.
\end{equation}
By an adaptation of the proof of Proposition 3.1, (P6), in \cite{abry2018waveleta} we obtain
\begin{equation}\label{e:B-hat^(tau)(2^j)=scaling}
\widehat{{\mathbf B}}^{(\tau)}_a(2^j) \stackrel{d}= \big((2^j)^{H_{\ell_1}+H_{\ell_2}+1}\hspace{0.5mm}\widehat{b}^{(\tau)}_{\ell_1,\ell_2}(1)\big)_{\ell_1,\ell_2 = 1,\hdots,r}.
\end{equation}
So, by \eqref{e:log2_sqrt(n_a,j)(fN,m(B-hat_a(2^j))-fN,m(B_a(2^j)))}, \eqref{e:B(2^j)=scaling}, \eqref{e:x*,m_reexpressed} and \eqref{e:B-hat^(tau)(2^j)=scaling}, as $N \rightarrow \infty$,
\begin{equation*}
\begin{aligned}
&\sqrt{n_{a,j^0_2}}\big\{ \log_2\hat \lambda_{m}^{(\tau)}(a(N)2^j)- \log_2 \lambda_{m}^{(\tau)}(a(N)2^j) \big\} \\
&\stackrel{d}\sim \frac{1}{\ln 2}\frac{\sqrt{n_{a,j^0_2}}}{\xi_{m}(2^j)}\Bigg\{\sum^{r_1+r_2}_{i_1 = r_1+1}\sum^{r_1+r_2}_{i_2 = r_1+1} {\boldsymbol \gamma}_{m,i_1}{\boldsymbol \gamma}_{m,i_2} \hspace{0.5mm}(2^j)^{2H_m+1}  \\
&\quad + 2 \sum^{r_1+r_2}_{i_1 = r_1+1}\sum^{M}_{i_2 = r_1+r_2+1}
 {\boldsymbol \gamma}_{m,i_1}  [(2^j)^{H_m}(2^j)^{-H_{i_2}}] \\
&\qquad  \quad {\boldsymbol \varrho}_{m,i_2} (2^j)^{H_m + H_{i_2}+1}\\
& \quad+\sum^{M}_{i_1 = r_1+r_2+1}\sum^{M}_{i_2 = r_1+r_2+1}
[(2^j)^{H_m}(2^j)^{-H_{i_1}}]{\boldsymbol \varrho}_{m,i_1} \\
 & \qquad \quad [(2^j)^{H_m}(2^j)^{-H_{i_2}}]{\boldsymbol \varrho}_{m,i_2}\hspace{0.5mm}(2^j)^{H_{i_1} + H_{i_2}+1}\hspace{0.5mm}\Bigg\} \\
&  \qquad \qquad \cdot \pi_{i_1,i_2}\big(\widehat{{\mathbf B}}^{(\tau)}_{a}(1) - {\mathbf B}(1)\big).
\end{aligned}
\end{equation*}
In view of the scaling relation in \eqref{e:lim_n_a*lambda/a^(2h+1)}, relation \eqref{e:Taylor_last_equivalence} holds, as claimed.
\end{remark}

\subsection{Approximations of log-eigenvalue covariances}
\label{sec:appendixC}

Fix $j \in \{j^0_1,\hdots,j^0_2\}$. In this appendix, we use the bivariate context to provide some justification for the very useful approximations 
\begin{align}
\label{eq:corrlambdaapprox}
& \mathrm{corr}(\log_2 \hat \lambda_m(a(N)2^j),\log_2 \hat \lambda_{m'}(a(N)2^j)) \approx 0, \\
\label{eq:corrlambdawapprox}
& \mathrm{corr}(\log_2 \hat \lambda_m^{(\tau)}(a(N)2^j),\log_2 \hat \lambda_{m'}^{(\tau)}(a(N)2^j)) \approx 0,
\end{align}
as well as
\begin{align}
\label{eq:varlambdaapprox}
& \Var(\log_2 \hat \lambda_m(a(N)2^j)) \approx 2(\log_2 e)^2/n_{a,j}, \\
\label{eq:varlambdawapprox}
& \Var(\log_2 \hat \lambda_m^{(\tau)}(a(N)2^j)) \approx 2(\log_2 e)^2/n_{a,j^0_2}.
\end{align}

So, consider the case $M = 2$ and $0 < H_1 < H_2 < 1$. Since the arguments for showing all relations \eqref{eq:corrlambdaapprox}-\eqref{eq:varlambdawapprox} are similar, calculations are detailed for the bias-corrected eigenvalues $ \hat \lambda_m^{(\tau)}$ only.
Let $\widehat{{\mathbf B}}_a(2^j)$ be as in \eqref{e:Bhat^(tau)} and further write entry-wise ${\mathbf B}(1) = \bbE \widehat{{\mathbf B}}_{a}(1) = \big\{b_{\ell_1 \ell_2}(1) \big\}_{\ell_1,\ell_2 = 1,2}$ (cf.\ \eqref{e:B(2^j)=(B_i,ell)}). Assume
\begin{equation}\label{e:b^2-12(1)>0}
b^{2}_{12}(1) > 0.
\end{equation}

In view of \eqref{e:Taylor_last_equivalence}, it can be shown that we can write the expansions
\begin{equation}\label{e:log_lambda1_expansion}
\begin{aligned}
& \sqrt{n_{a,j^0_2}} \big(\mathrm{ln} \hat \lambda^{(\tau)}_{1}(a(N)2^j) - \mathrm{ln} \lambda^{(\tau)}_{1}(a(N)2^j) \big) \\
& \quad \stackrel{d}\sim \frac{b_{11}(1)b_{22}(1)}{\det {\mathbf B}(1)} \sqrt{n_{a,j^0_2}} \cdot \frac{(\widehat{b}_{11}(1) - b_{11}(1))}{b_{11}(1)} \\
&  \qquad  + (-2) \frac{b^2_{12}(1)}{\det {\mathbf B}(1)} \sqrt{n_{a,j^0_2}} \cdot \frac{(\widehat{b}_{12}(1) - b_{12}(1))}{b_{12}(1)} \\
&  \qquad  + \frac{b^2_{12}(1)}{\det {\mathbf B}(1)} \sqrt{n_{a,j^0_2}} \cdot \frac{(\widehat{b}_{22}(1) - b_{22}(1))}{b_{22}(1)}
\end{aligned}
\end{equation}
and
\begin{equation}\label{e:log_lambda2_expansion}
\begin{aligned}
& \sqrt{n_{a,j^0_2}} \big(\mathrm{ln} \hat \lambda^{(\tau)}_{2}(a(N)2^j) - \mathrm{ln} \lambda^{(\tau)}_{2}(a(N)2^j) \big) \\
&\qquad  \qquad \qquad  \qquad \stackrel{d}\sim \sqrt{n_{a,j^0_2}} \cdot \frac{(\widehat{b}_{22}(1) - b_{22}(1))}{b_{22}(1)}
\end{aligned}
\end{equation}
(see relations (4.33) and (4.35) in \cite{abry2018waveleta}). 
So, for the sake of computation, we normalize $b_{11}(1) = 1 = b_{22}(1)$. Also, for notational simplicity, we write $X_{k} = \pi_1\big( \mathbf{W}^{-1}D(1,k)\big)$, $Y_{k} = \pi_2\big( \mathbf{W}^{-1}D(1,k)\big)$, $\forall k \in \integer$, where $\pi_1$ and $\pi_2$ denote the projection onto the $\ell$-coordinate. For ofBm, the celebrated property of approximate decorrelation in the wavelet domain is given by Proposition 3.2 in \cite{abry2018waveleta}. It implies that, for $k \neq k'$,
\begin{equation}\label{e:wavelet_decorr}
\bbE X_{k} X_{k'} \approx 0, \quad  \bbE Y_{k} Y_{k'} \approx 0, \quad \bbE X_{k} Y_{k'} \approx 0.
\end{equation}
In view of \eqref{e:wavelet_decorr}, of the stationarity of fixed-scale wavelet coefficients (see \cite{abry2018waveleta}, Proposition 3.1, $(P2)$) and of the Isserlis theorem, we obtain the approximations
\begin{equation}\label{e:Cov(b11,b22)}
\Cov\big(\widehat{b}_{11}(1),\widehat{b}_{22}(1)\big) \approx \frac{2 b^{2}_{12}(1)}{n_{a,j^0_2}},
\end{equation}
\begin{equation}\label{e:Cov(b22,b22)}
\Cov\big(\widehat{b}_{12}(1),\widehat{b}_{22}(1)\big)  \approx \frac{2 b_{12}(1)}{n_{a,j^0_2}}, \hspace{1mm}\Var\big(\widehat{b}_{22}(1)\big) \approx \frac{2}{n_{a,j^0_2}}.
\end{equation}
As a consequence of the expansions \eqref{e:log_lambda1_expansion} and \eqref{e:log_lambda2_expansion}, as well as of the approximations \eqref{e:Cov(b11,b22)} and \eqref{e:Cov(b22,b22)},
\begin{equation}
\begin{aligned}
 \Cov \Big(\sqrt{n_{a,j^0_2}} & \big(\mathrm{ln} \hat \lambda^{(\tau)}_{1}(a(N)2^j) - \mathrm{ln} \lambda^{(\tau)}_{1}(a(N)2^j)\big),\\
& \quad \sqrt{n_{a,j^0_2}} \big(\mathrm{ln} \hat \lambda^{(\tau)}_{2}(a(N)2^j) -\mathrm{ln} \lambda^{(\tau)}_{2}(a(N)2^j)\big) \Big) \\
\approx \frac{n_{a,j^0_2}}{\det {\mathbf B}(1)} & \Big\{ \frac{2 b^{2}_{12}(1)}{n_{a,j^0_2}} - 2 \cdot \frac{b^{2}_{12}(1)}{b_{12}(1)} \frac{2 b_{12}(1)}{n_{a,j^0_2}} + \frac{2 b^{2}_{12}(1)}{n_{a,j^0_2}}\Big\} = 0.
\end{aligned}
\end{equation}
Namely, $\log_2 \hat \lambda^{(\tau)}_{1}(a(N)2^j)$ and $\log_2 \hat \lambda^{(\tau)}_{2}(a(N)2^j)$ are approximately decorrelated over large scales for large values of $n_{a,j^0_2}$.

Alternatively, suppose condition \eqref{e:b^2-12(1)>0} is dropped, namely, assume $b^{2}_{12}(1) = 0$. Then, after appropriately eliminating null denominators, the expansions \eqref{e:log_lambda1_expansion} and \eqref{e:log_lambda2_expansion} and the approximation \eqref{e:Cov(b11,b22)} immediately show that the log-eigenvalues are approximately decorrelated.\vspace{2mm}

Next, we develop approximations for the variances. Similarly to \eqref{e:Cov(b11,b22)} and \eqref{e:Cov(b22,b22)}, we obtain
\begin{align}
\Cov\big(\widehat{b}_{11}(1),& \widehat{b}_{12}(1)\big) \approx \frac{2 b_{12}(1)}{n_{a,j^0_2}}, \\
\Var\big(\widehat{b}_{11}(1)\big) \approx \frac{2}{n_{a,j^0_2}}, & \hspace{2mm}\Var\big(\widehat{b}_{12}(1)\big)  \approx \frac{ b^2_{12}(1)}{n_{a,j^0_2}} + \frac{1}{n_{a,j^0_2}}.
\end{align}

Hence,
\begin{equation}
\begin{aligned}
& \Var \Big(\sqrt{n_{a,j^0_2}} \big(\mathrm{ln} \hat \lambda^{(\tau)}_{1}(a(N)2^j) - \mathrm{ln} \lambda^{(\tau)}_{1}(a(N)2^j)\big) \Big) \\
& \qquad \approx \frac{n_{a,j^0_2}}{\det {\mathbf B}(1)^2} \Bigg[  \Var \hat b_{11}(1) + 4 b_{12}^2(1) \Var \hat b_{12}(1) \\
& +  b_{12}^4(1) \Var \hat b_{22}(1)  - 4 b_{12}(1)  \Cov(\hat b_{11}(1),\hat b_{12}(1)) \\
& + 2 b_{12}(1)^2  \Cov(\hat b_{11}(1),\hat b_{22}(1)) - 4 b_{12}^3(1)  \Cov(\hat b_{12}(1),\hat b_{22}(1)) \Bigg] \\
&\qquad =\frac{n_{a,j^0_2}}{(1-b^2_{12}(1))^2} \Bigg[  \frac{2}{n_{a,j^0_2}} +4 b^2_{12}(1) \left( \frac{b_{12}^2(1)}{n_{a,j^0_2}} + \frac{1}{n_{a,j^0_2}} \right) \\
&\qquad \qquad  +  b_{12}^4(1) \frac{2}{n_{a,j^0_2}} -4 b_{12}(1) \frac{2 b_{12}(1)}{n_{a,j^0_2}}  +2 b_{12}^2(1)  \frac{2 b^{2}_{12}(1)}{n_{a,j^0_2}} \\
&\qquad \qquad - 4 b_{12}^3(1)  \frac{2 b_{12}(1)}{n_{a,j^0_2}} \Bigg] = 2\\
\end{aligned}
\end{equation}
and
\begin{equation}
\begin{aligned}
& \Var \Big(\sqrt{n_{a,j^0_2}} \big(\mathrm{ln} \hat \lambda^{(\tau)}_{2}(a(N)2^j) - \mathrm{ln} \lambda^{(\tau)}_{2}(a(N)2^j)\big) \Big) \\
& \qquad \approx n_{a,j^0_2} \cdot \frac{2}{n_{a,j^0_2}} = 2.
\end{aligned}
\end{equation}

Analogously, we can show that
\begin{equation}
\Var \Big(\sqrt{n_{a,j}} \big(\mathrm{ln} \hat \lambda_{2}(a(N)2^j) - \mathrm{ln} \lambda_{2}(a(N)2^j)\big) \Big) \approx 2.
\end{equation}

\begin{remark}
\label{remC}
Relations \eqref{eq:corrlambdaapprox}-\eqref{eq:varlambdawapprox} provide approximations for the second moments of wavelet log-eigenvalues that depend neither on model parameters nor on the scale. Beyond these approximations, Remark~\ref{remB} shows rigorously and more broadly that the limit in distribution \eqref{e:scale_invariance_of_the_distribution} of wavelet log-eigenvalues for a single $j$ is \textit{scale-free}.
\end{remark} 

\subsection{Additional covariance approximations for multivariate estimation}
\label{sec:appendixD}

By \eqref{eq:corrlambdaapprox}, pairs of multivariate estimates $\hat H^{\rm M}_m$ and $\hat H^{\rm M}_{m'}$, $m \neq m'$, are approximately decorrelated. On the other hand, by~\eqref{eq:corrlambdawapprox}, the same holds for pairs of their bias-corrected counterparts $\hat H^{\rm M,bc}_m$ and $\hat H^{\rm M,bc}_{m'}$. In summary, for $m \neq m'$,
$$
\mathrm{corr}(\hat H^{\rm M}_m,\hat H^{\rm M}_{m'}) \approx 0, \quad \mathrm{corr}(\hat H^{\rm M,bc}_m,\hat H^{\rm M,bc}_{m'}) \approx 0.
$$

On the other hand, by \eqref{eq:corrlambdawapprox}, we can further approximate
$$
\mathrm{Var}(2\hat H^{\rm M,bc}_m) = \mathrm{Var} \left( \sum_{j=j_1^0}^{j_2^0} w_j \overline{\log_2 \lambda_m}(a(N)2^j) \right)
$$
$$
\approx \sum_{j=j_1^0}^{j_2^0} w_j^2 \, \mathrm{Var} \left( \frac{1}{2^{j_2^0-j}} \sum_{\tau=1}^{2^{j_2^0-j}} \log_2 \hat \lambda_m^{(\tau)}(a(N)2^j) \right)
$$
$$
\approx  \sum_{j=j_1^0}^{j_2^0}  \frac{w_j^2}{2^{j_2^0-j}} \mathrm{Var} \left(  \log_2 \hat \lambda_m^{(1)}(a(N)2^j) \right)
$$
\begin{equation}
\approx  2 (\log_2 e)^2 \sum_{j=j_1^0}^{j_2^0}  2^{j-j_2^0}  \frac{w_j^2}{n_{a,j_2^0}}
=  \displaystyle 2(\log_2 e )^2 \sum_{j=j_1^0}^{j_2^0} \frac{w_j^2}{n_{a,j}},
\end{equation}
where the last line results from $n_{a,j} 2^j =  n_{a,j_2^0} 2^{j_2^0} $.

Similarly, by \eqref{eq:corrlambdaapprox},
\begin{align}
\mathrm{Var}(2 \hat H^{\rm M}_m) & \approx  2 (\log_2 e)^2 \sum_{j=j_1^0}^{j_2^0}   \frac{w_j^2}{n_{a,j}}. 
\end{align}

 \end{document}